\documentclass{article}
\usepackage{amsthm, amssymb, amsmath}
\usepackage{amsthm}
\usepackage{graphicx}
\numberwithin{equation}{section}
\newtheorem{theorem}{Theorem}[section]

\newtheorem{lemma}[theorem]{Lemma}

\newtheorem{remark}{Remark}
\newtheorem{proposition}[theorem]{Proposition}

\huge

\newcommand{\lb}{\left(}
\newcommand{\rb}{\right)}

\newcommand{\nc}{\newcommand}

\nc{\be}{\begin{equation}}
\nc{\la}{\label}
\nc{\ba}{\begin{array}}
\nc{\ea}{\end{array}}
\nc{\bs}{\begin{split}}
\nc{\es}{\end{split}}

\nc{\J}{\mathbb J}

\nc{\pt}{\partial_t}
\nc{\ptt}{\partial_t^2}
\nc{\e}{\epsilon}
\nc{\lam}{\lambda}
\nc{\G}{\Gamma}
\nc{\g}{\gamma}
\nc{\al}{\alpha}
\nc{\del}{\delta}
\nc{\Om}{\Omega}

\newcommand{\ra}{\rightarrow}

\date{}

\newcommand{\DETAILS}[1]{}

\begin{document}
\title{Friction in a Model of Hamiltonian Dynamics}
\author{J\"urg Fr\"ohlich\footnote{juerg@itp.phys.ethz.ch},\  Zhou Gang\footnote{zhougang@itp.phys.ethz.ch} and Avy Soffer\footnote{soffer@math.rutgers.edu}}
\maketitle
\setlength{\leftmargin}{.1in}
\setlength{\rightmargin}{.1in}
\normalsize \vskip.1in
\setcounter{page}{1} \setlength{\leftmargin}{.1in}
\setlength{\rightmargin}{.1in}
\large
\centerline{$^{\ast,\dagger}$Institute for Theoretical Physics, ETH Zurich, CH-8093, Z\"urich, Switzerland}
\centerline{$^{\ddagger}$Department of Mathematics, Rutgers University, New Jersey 08854, USA}
\date

\setlength{\leftmargin}{.1in}
\setlength{\rightmargin}{.1in} 

\normalsize \vskip.1in
\setcounter{page}{1} \setlength{\leftmargin}{.1in}
\setlength{\rightmargin}{.1in}
\large
{\it{``A moving body will come to rest as soon as the force pushing it no longer acts on it in the manner necessary for its propulsion."}}------ Aristotle
\section*{Abstract}
We study the motion of a heavy tracer particle weakly coupled to a dense ideal Bose gas exhibiting Bose-Einstein condensation. In the so-called mean-field limit, the dynamics of this system approaches one determined by nonlinear Hamiltonian evolution equations describing a process of emission of Cerenkov radiation of sound waves into the Bose-Einstein condensate along the particle's trajectory. The emission of Cerenkov radiation results in a friction force with memory acting on the tracer particle and causing it to decelerate until it comes to rest.
\tableofcontents

\section{Introduction}
On the basis of abundant everyday experience, Aristotle formulated a naive dynamical law, which, roughly speaking, says that the velocity of a moving body is proportional to the external force acting on it and that it approaches a state of rest as soon as the force that has propelled it no longer acts on it. This law captures the phenomenon of friction. We have learnt from the discoveries of Galileo, Newton and their followers that Aristotle's law is not the right starting point for the development of analytical mechanics. It is the acceleration rather than the velocity of a body moving in empty space that is proportional to the force acting on it. Yet, particle motion with friction caused by dissipative processes is an omni-present phenomenon. It is thus of considerable interest to analyze how this type of motion emerges from the Hamiltonian dynamics of particles coupled to a dispersive environment.

The problem of constructing the effective dynamics of a particle coupled to dispersive reservoirs has previously been studied, e.g. in ~\cite{CaMaPu,BrLa} and references given there. In ~\cite{FGSS}, we have introduced a family of quantum mechanical models describing tracer particles moving through a Bose gas exhibiting Bose Einstein condensation. We have identified a regime in which the dynamics of this system approaches one governed by classical Hamiltonian evolution equations. Similar equations have also been considered in \cite{KM}. This so-called ``mean field regime" is characterized as follows:
\begin{itemize}
\item[(i)] The mass of the tracer particle in chosen to be $M_{p}=NM,\ M>0$, and the potential of external forces acting on it is $V_{p}(X)=NV(X),$ where $X$ is the particle position, and $N$ is a parameter ranging over the interval $[1,\infty).$
\item[(ii)] The mean density of the Bose gas is chosen to be $\rho=N\frac{\rho_0}{g^2}$, $\rho_0>0,$ where $g$ is a real parameter. The coupling constant of two-body forces between atoms in the Bose gas is chosen to be $\lambda=\frac{\kappa}{N},\ \kappa\geq 0.$ The  two-body forces are derived from a potential $\phi(x-y)$ assumed to be {\bf{spherically symmetric}}, of {\bf{short range}} and of {\bf{positive type}}, (i.e., the two-body force is repulsive in average); $x$ and $y$ denote the positions of two atoms in the Bose gas. The mass of an atom is denoted by $m$.
\item[(iii)] The interaction between the tracer particle and an atom in the Bose gas is described by a two-body potential $gW(x-X)$, where the coupling constant $g$ is the same parameter as the one introduced in (ii) and $W$ is {\bf{spherically symmetric}} and of {\bf{short range}}.
\end{itemize}

The {\bf{mean-field regime}} corresponds to the limit
\begin{align}\label{eq:meanFL}
N\rightarrow \infty.
\end{align}

\begin{remark}
If the two-body Schr\"odinger operator $-\frac{1}{2m}\Delta_{x}+gW(x)$ has bound states and the coupling constant $\kappa$ is strictly positive then the ``effective mass" of the tracer particle can be argued to be proportional to $Nm$, because it forms a bound state with $O(N)$ atoms of the Bose gas. In this situation, the assumption that $M_{p}=NM$, see (i), is presumably superfluous; see \cite{FGSS}.
\end{remark}

Heuristic arguments (see ~\cite{FGSS}), which can be made rigorous, mathematically, for ideal Bose gases $(\kappa=0)$ indicate that, in the mean-field limit~\eqref{eq:meanFL}, the dynamics of the system is described by the following classical, nonlinear Hamiltonian evolution equations:
\begin{align}
\dot{X_t}=&\frac{P_t}{M},\quad\quad
\dot{P_t}=-\nabla_{X}V(X_t)+g\int\nabla_{x}W(X_t-x)\lb|\alpha_t(x)|^2
-\frac{\rho_0}{g^2}\rb dx,    \label{XPeqns2}\\
i\dot{\alpha}_t(x)=&\lb-\frac{1}{2m}\Delta+gW(X_t-x)\rb\alpha_t(x) 
+\kappa\phi *\lb|\alpha_t(y)|^2-\frac{\rho_0}{g^2}\rb\alpha_t(x).              \label{alphaeqn}
\end{align}
In Eqs. ~\eqref{XPeqns2} and ~\eqref{alphaeqn}, $X_t\in \mathbb{R}^3$ and $P_{t}\in \mathbb{R}^3$ are the position and momentum of the tracer particle, respectively, at time $t$, and $\alpha_t(x)$ is the Ginzburg-Landau order-parameter field describing the state of the Bose gas at time $t$ (in the mean-field limit). The interpretation of $|\alpha_{t}(x)|^2$ is that of the density of atoms at the point $x$ of physical space $\mathbb{R}^3$, at time $t;$ the global phase of $\alpha_t$ is not an observable quantity.
The symbol $*$ in ~\eqref{alphaeqn} denotes convolution.

We impose the conditions that $\nabla \alpha_t$ is square-integrable in $x$ and that $|\alpha_t|^2-\frac{\rho_0}{g^2}$ is integrable. This defines an affine space of complex-valued functions denoted by $\Gamma_{BG}$ whose tangent space can be chosen to be some weighted Sobolev space. We define $\Gamma$ to be the Cartesian product of $\mathbb{R}^6$ (the tracer particle's phase space) with $\Gamma_{BG}$. The space $\Gamma$ is the phase space of the system. It is equipped with the standard symplectic form
\begin{align}
\omega=\sum_{n=1}^{3}dX^{n}\wedge dP_{n}-i \int d\bar\alpha\wedge d\alpha.
\end{align}
Eqs. ~\eqref{XPeqns2} and ~\eqref{alphaeqn} then turn out to be the Hamiltonian equations of motion corresponding to the Hamilton functional
\begin{align}\label{Ham}
 H(X,P,\alpha,\bar{\alpha})&=&\frac{P^2}{2M}+V(X)+\int dx \{
\frac{1}{2m}|\nabla\alpha(x)|^2+gW(X-x)\ (|\alpha(x)|^2-\frac{\rho_0}{g^2})\ \}dx\\
& &+\frac{\kappa}{2}\int \int dx dy
\ (|\alpha(y)|^2-\frac{\rho_0}{g^2})\ \phi(y-x)
\ (|\alpha(x)|^2-\frac{\rho_0}{g^2}).\nonumber
\end{align}
This functional is sometimes called ``Gross-Pitaevskii functional." It exhibits a global $U(1)$ symmetry,
\begin{align}
\alpha(x)\rightarrow e^{i\theta} \alpha(x),\ \ \bar{\alpha}(x)\rightarrow e^{-i\theta}\alpha(x),
\end{align} where $\theta$ is an $x-$independent phase. This symmetry is broken by a choice of boundary conditions at $\infty.$ In this paper, we choose zero-vorticity boundary conditions,
\begin{align*}
\alpha(x)\xrightarrow[|x|\rightarrow \infty] {} \sqrt{\frac{\rho_0}{g^2}}.
\end{align*}
We introduce a ``fluctuation field", $\beta$, by setting
\begin{equation} \label{alphabetarel}
\alpha(x):=\sqrt{\frac{\rho_0}{g^2}}+\beta(x),
\end{equation} with $\beta(x)\rightarrow 0$, as $|x|\rightarrow \infty.$ The equations of motion then read
\begin{align}
\dot{X_t}=&\frac{P_t}{M},\quad\quad
\dot{P_t}=-\nabla_{X}V(X_t)+g\int\nabla_{x}W(X_t-x)\lb|\beta_t(x)|^2
+2\sqrt{\frac{\rho_0}{g^2}}Re\beta_t(x)\rb dx,    \label{XPeqns}\\
i\dot{\beta}_t(x)=&\lb-\frac{1}{2m}\Delta+gW(X_t-x)\rb
\beta_t(x)+\sqrt{\rho_0}W(X_t-x)                  \nonumber\\
+&\kappa\lb\phi *\lb|\beta_t|^2+2\sqrt{\frac{\rho_0}{g^2}}
Re\beta_t\rb\rb(x)\lb\beta_t(x)+\sqrt{\frac{\rho_0}{g^2}}\rb. \label{betaeqn}
\end{align}
The Hamilton functional giving rise to these equations is obtained from ~\eqref{Ham} after inserting the substitution
~\eqref{alphabetarel}. Eqs. ~\eqref{XPeqns} and ~\eqref{betaeqn} have stationary (time-independent) solutions, and if the external force vanishes ($V\equiv 0$) they have traveling wave solutions, provided the speed of the particle is smaller than or equal to the speed of sound in the condensate; see ~\cite{FGSS}. If the initial speed of the particle is larger than the speed of sound of the condensate a non-zero friction force is generated, because the particle emits sound waves into the condensate (Cerenkov radiation) and hence loses energy until its speed drops to the speed of sound, whereupon it continues to move ballistically, accompanied by a ``splash" in the order-parameter field $\beta$. (Quantum mechanically, this splash corresponds to a coherent states of gas atoms and causes decoherence in particle-position space, allowing for an essentially ``classical" detection of the particle trajectory.)

The following models are of interest; see ~\cite{FGSS}:
\begin{itemize}
\item[B]-Model: $\kappa=0$ (ideal Bose gas), $g\rightarrow 0,$ see \cite{FGS}.
\item[C]-Model: $\kappa=0$ and $g\neq 0;$ see ~\cite{EG}.
\item[E]-Model: $\kappa\rho_0/g^2 =\mu=$const., and $g, \kappa\ra  0$
\item[G]-Model: $\kappa>0$ and $g\neq 0$.
\end{itemize}

The $B-$ model is a special case of the E-model ($\mu=0$). The C-model is a little harder to analyze than the B-model, and one must assume that the operator $-\frac{1}{2m}\Delta+gW$ does not have bound states or zero-energy resonances. (Bound states would cause an instability in the $C-$ model and would presumably lead to ``run-away" solutions; see ~\cite{FGSS}). Work on the $E-$ model is in progress.

For further discussions of the physics background of these models, special solutions of the equations of motion and references to other studies of related problems we refer to ~\cite{FGSS}.

In this paper, we focus our attention on the simplest model, the $B-$model, with $V\equiv 0.$ The equations of motion are then given by
\begin{align}
\dot{X}_{t}=\frac{P_{t}}{M},\ \ &
\dot{P}_{t}=2\sqrt{\rho_0}Re\langle \nabla_{x} W^{X_{t}},\beta_{t}\rangle;\label{eq:preDF}\\
i\dot\beta_{t}(x)=&-\frac{1}{2m}\Delta \beta_{t}+\sqrt{\rho_{0}}W^{X_{t}},\label{eq:preDF2}
\end{align} where
\begin{align}
W^{X}(x):=W(X-x), \ \langle f,\ g\rangle:=\int \bar{f}(x) g(x)\ dx.
\end{align} The corresponding Hamilton functional is given by
\begin{align}\label{eq:conserv}
H(X, P;\ \bar\beta,\beta):=\frac{|P|^{2}}{2M}+\frac{1}{2m}\int_{\mathbb{R}^3} |\nabla \beta|^2 dx+2\sqrt{\rho_0}\int_{\mathbb{R}^3}W^{X}Re \beta dx.
\end{align}

The main result established in this paper says that if the initial kinetic energies of the particle and of the fluctuations in the Bose gas, as described by $\beta$, are small enough, as compared to $\rho_0\int dx \ W(x),$ (and if $\beta_{t=0}$ decays sufficiently rapidly at $\infty$) then the solutions of ~\eqref{eq:preDF}, ~\eqref{eq:preDF2} have the properties that
\begin{align}
&|P_t|\searrow 0,\  \text{integrably fast in time}\ t,\nonumber\\
&X_t\rightarrow X_\infty,\ \text{and}, \label{eq:resu}\\
&\beta_t\rightarrow 2m\sqrt{\rho_0} \Delta^{-1} W^{X_{\infty}},\ \text{as}\ t\rightarrow \infty.\nonumber
\end{align}
By scaling time $t,$ space $x$ and the fluctuation field $\beta$, denoting the new variables again by $t,$ $x$ and $\beta$, we can write the equations of motion of the $B-$ model in the form
\begin{align}
\dot{X}_{t}=P_{t},\ \ &
\dot{P}_{t}=\nu Re\langle \nabla_{x} W^{X_{t}},\beta_{t}\rangle;\label{eq:DF}\\
i\dot\beta_{t}(x)=&-\frac{1}{2}\Delta \beta_{t}+W^{X_{t}},\label{eq:DF2}
\end{align} with
\begin{align}\label{eq:scale}
|\hat{W}(0)|=(2\pi)^{-\frac{3}{2}}|\int d^3x\ W(x)|=1.
\end{align}
Henceforth, we study ~\eqref{eq:DF} and ~\eqref{eq:DF2}, with the normalization condition ~\eqref{eq:scale} imposed.

Technically, the main difficulty surmounted in our paper is to construct solutions of a certain semi-linear integro-differential equation, (see ~\eqref{eq:linear} below) whose linearization (in $P$) takes the form
\begin{align}\label{linearPa}
\partial_{t}q_{t}=Z Re\langle W, \  e^{i\frac{\Delta}{2}t} W\rangle \int_{0}^{t} q_{s}ds-Z Re \langle W,\ \int_{0}^{t}  e^{i\frac{\Delta}{2}(t-s)} \ q_{s}\ ds\ W\rangle
\end{align} where $q_t$ is a component of $P_t$, and $Z>0$ is a positive constant. Using properties of the solution of ~\eqref{linearPa}, we reformulate the original semi-linear integro-differential equation in such a way that a fixed-point theorem becomes applicable to construct solutions, provided that the initial conditions are small enough.

Our paper is organized as follows: in Section ~\ref{SEC:main} we carefully state our main result, Theorem ~\ref{THM:main}. In Section ~\ref{sec:Refor} we rewrite ~\eqref{eq:DF} and ~\eqref{eq:DF2} in a more convenient form. This leads to an equation for $P_{t}$ containing a linear part and a higher-order nonlinear part. The linearized equation is then carefully studied in Section ~\ref{Sec:RefTHM}, which, technically, is the core of the present paper. The proof of our main result is completed in Section ~\ref{sec:ProofMainTheorem}. In several appendices (Appendices ~\ref{Sec:ReforPro} through ~\ref{sec:Kideas}) some auxiliary technical results are established.
\section*{Acknowledgements}
We are indebted to I.M.Sigal for highly stimulating discussions on the problems solved in this paper and very useful suggestions. We also thank D.Egli for very helpful observations.
\section{The Main Theorem}\label{SEC:main}
In this section, we summarize our main results concerning the solutions of the equations of motion ~\eqref{eq:DF} and ~\eqref{eq:DF2}.
For this purpose, we define a certain interval $I\in (0,1)$ of real numbers, $\delta,$ as follows
\begin{equation}\label{eq:difI}
I:=\{\delta\ |\delta>0, \ \int_{0}^{1}\frac{1}{1+(1-r)^{\frac{1}{2}}}(1-r)^{-\frac{1}{2}} [\frac{1}{1-2\delta}(r^{-\frac{1}{2}}-r^{-\delta})+ r^{\frac{1}{2}-\delta}]\ dr<\pi
\}.
\end{equation}
Numerical evaluation of the integral on the right hand side of ~\eqref{eq:difI} on a computer shows that $I$ is non-empty, with $$I_{sup}:=\sup_{\delta}\{\delta|\ \delta\in I\}\approx 0.66.$$

The following theorem is the main result established in this paper.
\begin{theorem}\label{THM:main} Suppose that the two-body potential $W$ in ~\eqref{eq:conserv}, ~\eqref{eq:DF} and ~\eqref{eq:DF2} is smooth, spherically symmetric and of fast decay at infinity with $(2\pi)^{-\frac{3}{2}}|\int\ dx \ W(x)|=1,$ and that the constant $\nu$ in ~\eqref{eq:DF} and ~\eqref{eq:DF2} is $O(1)$. Then, given any $\delta\in I,$ there exists $\epsilon_0=\epsilon_0(\delta)$ such that if $$\|\langle x\rangle^{4}\beta_0\|_2,\ |P_0|\leq \epsilon_0,$$ then
\begin{equation}\label{eq:trajectory}
|P_{t}|\leq ct^{-\frac{1}{2}-\delta}\ \text{as}\ t\rightarrow \infty,
\end{equation} for some constant $c=c(\delta,\epsilon_0)<\infty.$ Moreover,
\begin{equation}\label{eq:convergence}
\lim_{t\rightarrow \infty}\| \beta_{t} +2(-\Delta)^{-1}W^{X_{t}}\|_{\infty}=0.
\end{equation}
\end{theorem}
The main theorem will be proven in Section ~\ref{sec:ProofMainTheorem}. In Section ~\ref{sec:Refor}, we derive a law/equation describing the effective dynamics of the tracer particle after eliminating the degrees of freedom of the Bose gas, and, in Section ~\ref{Sec:RefTHM}, we study the effective particle dynamics and, in particular, analyze its linearization, which represents a key ingredient of our analysis. Some heuristic arguments indicating what the true decay of the particle momentum $P_{t}$ in time $t$ might be are presented in Appendix ~\ref{SEC:bestconstants}.

\section{Effective Particle Dynamics, and Local Wellposedness}\label{sec:Refor}
To begin with, we recast equations ~\eqref{eq:DF} and ~\eqref{eq:DF2} in a form amenable to precise mathematical techniques.
We start with transforming Eq.~\eqref{eq:DF} into a convenient form. We introduce a new field, $\delta_t$, by
\begin{equation}\label{eq:decom}
\beta_{t}=:2\Delta^{-1}W^{X_t}+\delta_t.
\end{equation}
The first term on the right hand side of ~\eqref{eq:decom} describes a ``splash" in the Bose gas (a depletion if $W$ is repulsive, and an accretion if $W$ is attractive), while the field $\delta_t$ describes the emission of sound waves into the Bose gas by the tracer particle. In terms of the field $\delta_t$ the equations of motion ~\eqref{eq:DF} and ~\eqref{eq:DF2} are seen to be
\begin{align}
\dot{X}_{t}&=P_{t},\ \dot{P}_{t}=\nu Re\langle  \nabla_{x} W^{X_{t}},\  \delta_{t}\rangle \label{eq:effective}\\
i\dot\delta_{t}(x)&=-\frac{1}{2}\Delta\delta_{t}+ 2i \Delta^{-1}P_{t}\cdot \nabla_{x} W^{X_{t}}\nonumber
\end{align}
with
\begin{align}
\delta_0=-2\Delta^{-1}W^{X_0}+\beta_0,
\end{align}
where, in contrast to the term $-2\Delta^{-1}W^{X_0}$, $\beta_0$ has good decay at infinity. Using Duhamel's principle, we obtain for $\delta_t$
\begin{align}\label{eq:Duh}
\delta_{t}
=-2  e^{i\frac{\Delta}{2}t}  (\Delta)^{-1}W^{X_0}+ e^{i\frac{\Delta}{2}t} \beta_0+2\int_{0}^{t}  e^{i\frac{\Delta}{2}(t-s)} (\Delta)^{-1} P_{s}\cdot \nabla_{x} W^{X_{s}}\ ds.
\end{align}
Plugging this equation into the equation for $P_{t},$ we find that
\begin{align}
\dot{P}_{t}&=\nu Re\langle \nabla_{x} W^{X_{t}},  \delta_{t}\rangle\nonumber\\
&=\nu Re\langle \nabla_{x} W^{X_{t}}, \  e^{i\frac{\Delta}{2}t} \beta_0\rangle\nonumber\\
& +2\nu Re\langle \nabla_{x} W,\  e^{i\frac{\Delta}{2}t} (-\Delta)^{-1}W^{X_{0}-X_{t}}\rangle\label{eq:effective2}\\
& -2\nu Re\langle \nabla_{x} W, \int_{0}^{t}  e^{i\frac{\Delta}{2}(t-s)} (-\Delta)^{-1}P_s\cdot\nabla_{x} W^{X_{s}-X_{t}}\rangle ds,\nonumber
\end{align} where we have used that $\langle f^{X},\ g^{Y}\rangle=\langle f,\ g^{Y-X}\rangle.$
Next we use that
$$W^{Y-X_{t}}=W^{Y-X_0}-\int_{0}^{t} P_{s}\cdot \nabla_{x} W^{Y-X_s}\ ds,$$ to arrive at
\begin{align*}
&  Re\langle \nabla_{x} W,\  e^{i\frac{\Delta}{2}t} (-\Delta)^{-1}W^{X_{0}-X_{t}}\rangle\\
&=  Re\langle \nabla_{x} W,  e^{i\frac{\Delta}{2}t} (-\Delta)^{-1}[W^{X_0-X_{t}}-W]\rangle\\
&= Re\langle \nabla_{x} W,  e^{i\frac{\Delta}{2}t} (-\Delta)^{-1}\int_{0}^{t} P_{s} \ \cdot\nabla_x W\ ds\rangle\\
&
 +Re\langle \nabla_{x} W,  e^{i\frac{\Delta}{2}t} (-\Delta)^{-1}\int_{0}^{t} P_{s}  \ \cdot\nabla_x [W^{X_{0}-X_{s}}-W]\rangle\ ds;
\end{align*}
and
\begin{align*}
& Re\langle \nabla_{x} W, \int_{0}^{t}  e^{i\frac{\Delta}{2}(t-s)} (-\Delta)^{-1}P_s\cdot\nabla_x W^{X_{s}-X_{t}}\rangle ds\\
&=Re\langle \nabla_{x} W, \int_{0}^{t}  e^{i\frac{\Delta}{2}(t-s)} (-\Delta)^{-1}P_s\cdot\nabla_x W\rangle ds\\
& +Re\langle \nabla_{x} W, \int_{0}^{t} e^{i\frac{\Delta}{2} (t-s)}(-\Delta)^{-1}P_s\cdot\nabla_x  [W^{X_{s}-X_{t}}-W]\rangle\ ds.
\end{align*}
Plugging these identities into ~\eqref{eq:effective2}, we find that
\begin{align}
\dot{P}_{t}=F_{P}&+2\nu Re\langle \nabla_{x} W,  e^{i\frac{\Delta}{2}t} (-\Delta)^{-1}\int_{0}^{t} P_{s}\ ds \
\cdot\nabla_x W\rangle\nonumber\\
&-2\nu Re\langle \nabla_{x} W, \int_{0}^{t} e^{i\frac{\Delta}{2} (t-s)}(-\Delta)^{-1}P_s\cdot\nabla_x W\rangle\ ds,\label{eq:semilinear}
\end{align} where
\begin{equation}\label{eq:difVecF}
F_{P}:=B_0+B_{1}+B_{2}
\end{equation}
with $$B_0:=\nu Re\langle \nabla_{x} W^{X_{t}}, \  e^{i\frac{\Delta}{2}t} \beta_0\rangle,$$
$$B_{1}:=2\nu Re\langle \nabla_{x} W,  e^{i\frac{\Delta}{2}t} (-\Delta)^{-1}\int_{0}^{t} P_{s}  \ \cdot\nabla_x [W^{X_{0}-X_{s}}-W]\rangle\ ds$$ and
$$B_{2}:=-2\nu Re\langle \nabla_{x} W, \int_{0}^{t} e^{i\frac{\Delta}{2} (t-s)}(-\Delta)^{-1}\ P_s\cdot\nabla_x [W^{X_{s}-X_{t}}-W]\rangle\ ds_{1}ds,$$ the time argument, $t,$ being omitted in $B_i=B_i(t)$. Note that only $B_0$ depends on the initial condition, $\beta_0,$ of the Bose gas.

Equation ~\eqref{eq:semilinear} can be simplified somewhat as follows. Using that $W$ is spherically symmetric, we have that $$\langle \partial_{k}W,  e^{i\frac{\Delta}{2}t} (-\Delta)^{-1} P_{s} \ \cdot\nabla_x W\rangle=\langle \partial_{k}W,  e^{i\frac{\Delta}{2}t} (-\Delta)^{-1} p_{s}^{(k)} \partial_{k}W\rangle=\frac{1}{3}p_{s}^{(k)} \langle W,\  e^{i\frac{\Delta}{2}t} W\rangle,\ k=1,2,3,$$ where $\partial_{k}:=\frac{\partial}{\partial {x_{k}}},$ $P_s=(p_{s}^{(1)},p_{s}^{(2)}, p_{s}^{(3)})\in \mathbb{R}^3.$

Equation ~\eqref{eq:semilinear} is then seen to be equivalent to the following equation (or law)
\begin{align}\label{eq:linear}
\dot P_{t}=L(P)(t)+F_{P}(t)
\end{align} where $L$ is a linear operator on the space of momentum trajectories $\{P_s\}_{0\leq s<\infty}$ given by $$L(P):=\left(
L(p^{(1)}),\
L(p^{(2)}),\
L(p^{(3)})
\right)$$ with
\begin{align}
L(h)(t):=\frac{2}{3}\nu Re\langle W, \  e^{i\frac{\Delta}{2}t} W\rangle \int_{0}^{t} h_{s}\ ds-
\frac{2}{3}\nu Re \langle W,\ \int_{0}^{t}  e^{i\frac{\Delta}{2}(t-s)} h_{s}\ ds\ W\rangle,
\end{align} for an arbitrary function
$h:\mathbb{R}^{+}\rightarrow \mathbb{R}.$ Equation ~\eqref{eq:linear} is the law describing the effective dynamics of the tracer particle.

Next, we study the well-posedness of Equation ~\eqref{eq:linear}. Our results are summarized in the following theorem.
\begin{theorem}\label{THM:wellposed}
Equation ~\eqref{eq:linear} is locally well-posed: for $P_{0}\in \mathbb{R}^{3}$ and $\langle x\rangle^{3}\beta_{0}\in L^{2},$ there exists a positive time $T(|P_0|, \|\langle x\rangle^{3}\beta_0\|_{2})$ such that a solution $P_t$ of Equation ~\eqref{eq:linear} exists for any time $t$, with $0<t\leq T.$

In particular, for an arbitrary $T>0,$ there exists a constant $\epsilon_{0}(T)$ such that if $|P_0|$, $\|\langle x\rangle^{3}\beta_{0}\|_{2}\leq \epsilon_{0}(T)$ then $P_t$ is bounded by
\begin{equation}\label{eq:FiniInte}
|P_{t}|\leq T^{-2},\ \text{for any}\ t\in [0,T].
\end{equation}
\end{theorem}
\begin{proof}
The local well-posedness of Equation ~\eqref{eq:linear} is proven by standard techniques: One converts ~\eqref{eq:linear} into an integral equation for $P_t$ that can be solved by iteration, as long as time $t$ is small enough. The second part of Theorem ~\ref{THM:wellposed} follows from the (proof of the) first part and the observation that if $P_{0}=0$ and $\beta_0=0$ then $P_{t}=0$ is a global solution to ~\eqref{eq:linear}.
\end{proof}


\section{Global Solution of ~\eqref{eq:linear}}\label{Sec:RefTHM}
In this section we establish {\bf{global well-posedness}} of ~\eqref{eq:linear}.
We transform ~\eqref{eq:linear} to a more convenient form (see Equation ~\eqref{eq:secQt2}, below). For the new equation, we prove global existence of solutions by applying a fixed-point theorem on a suitably chosen Banach space of momentum trajectories.

To arrive at the new form of ~\eqref{eq:linear}, we introduce a ``propagator", $K$, solving the following Wiener-Hopf equation.
\begin{align}\label{eq:DIEkt}
\dot{K}=- Z Re \langle W,\ \int_{0}^{t}  e^{i\frac{\Delta}{2}(t-s)} K(s)\ ds\ W\rangle
\end{align} with
\begin{align*}
K(0)=1.
\end{align*} The constant $Z$ is given by $Z:=\frac{2}{3}\nu,$ where $\nu>0$ is the constant appearing in Equations ~\eqref{eq:DF} and ~\eqref{eq:effective}. Some properties of the function $K(t)$ are described in the following proposition.
\begin{proposition}\label{Prop:kernel}
The function $K:\mathbb{R}^{+}\rightarrow \mathbb{R}$ satisfies
\begin{equation}\label{eq:estKt}
Z K(t)=\frac{1}{4} \pi^{-\frac{5}{2}}   t^{-\frac{1}{2}}+C_{K} t^{-1}+O(t^{-\frac{3}{2}}),
\end{equation} as $t\rightarrow \infty,$ for some constant $C_{K}\in \mathbb{R}.$
\end{proposition}
A detailed proof of this proposition forms the contents of Appendix ~\ref{Sec:ReforPro}. On a formal level, the idea of the proof is straightforward. By taking Fourier transforms of both sides of ~\eqref{eq:DIEkt} in the time variable $t$ one obtains that
\begin{align}\label{eq:KtFourier}
\int_{0}^{\infty} e^{ikt} K_{t}\ dt=-\frac{1}{ik+Z G(k+i0)}
\end{align}
where the function $G:\mathbb{R}\rightarrow \mathbb{C}$ is given by
\begin{align}\label{eq:difGKi0}
G(k+i0):=-i\langle (-\Delta+2k+i0)^{-1}W, W\rangle +i\langle (-\Delta-2k-i0)^{-1}W, W\rangle.
\end{align} Thus,
\begin{align*}
K_{t}
&=-(2\pi)^{-1}\int_{-\infty}^{\infty} \frac{1}{ik+Z G(k+i0)} e^{-ikt}\ dk\\
&=-\frac{1}{\pi}\int_{-\infty}^{\infty} Re\frac{1}{ik+Z G(k+i0)}\ cos kt\ dk
\end{align*}
The function $G(k+i0)$ is smooth on $\mathbb{R}\backslash\{0\}$; in a neighborhood of $k=0$ we have that $G(k+i0)$ is an analytic function of $k^{\frac{1}{2}}$, and there exists a constant $C\not=0$ such that
\begin{align}\label{eq:TaylorGk}
G(k+i0)=C k^{\frac{1}{2}}+O(|k|)
\end{align}
Formula ~\eqref{eq:estKt} follows by finding the exact value for $C$.

Next, we rewrite Equation ~\eqref{eq:linear} in the following form
\begin{equation}\label{eq:Durhamel}
P_{t}=K(t) P_0+Z\int_{0}^{t}K(t-s) Re \langle W, e^{i\frac{\Delta}{2}  s}W\rangle \int_{0}^{s} P_{s_{1}} ds_{1}\ ds+\int_{0}^{t}K(t-s) F_{P}(s)\ ds
\end{equation}

In order to prove our main result on the decay of $P_t$ in $t$, see inequality ~\eqref{eq:trajectory} of Theorem ~\ref{THM:main}, we require more precise information on $P_t$ than the one provided by ~\eqref{eq:Durhamel}. The problem is that all three terms on the right hand side of ~\eqref{eq:Durhamel} are $O(t^{-\frac{1}{2}}),$ as $t\rightarrow \infty.$ This decay is inadequate to prove inequality ~\eqref{eq:trajectory}. We thus have to exhibit cancelations between the terms on the right hand side of ~\eqref{eq:Durhamel} that kill the leading terms. Our strategy to accomplish this is to resort to a second equation for $P_t$ equivalent to ~\eqref{eq:Durhamel} and then take a suitable linear combination of the two equations. To find the second equation we simply integrate both sides of ~\eqref{eq:linear} over time from $0$ to $t$ and arrive at
\begin{align}
P_{t}
=&P_{0}+Z\int_{0}^{t} Re\langle W,\ e^{i\frac{\Delta}{2}  s}W\rangle \int_{0}^{s} P_{s_{1}} \ ds_{1} ds-
2Z Re \langle W, (i\Delta)^{-1} \int_{0}^{t}  e^{i\frac{\Delta}{2}(t-s)}  P_{s} \ ds\ W\rangle\label{eq:integrate}\\
& +\int_{0}^{t}\ F_{P}(s)\ ds,\nonumber
\end{align} where the third term on the right hand side of ~\eqref{eq:integrate} derives from
$$
-Z\int_{0}^{t} Re\langle W, \int_{0}^{s} e^{i\frac{\Delta}{2}  (s-s_1)} q_{s_{1}}\ ds_{1}W\rangle ds
$$ by integrating by parts and using that $Re \langle W, (i\Delta)^{-1}W\rangle=0.$

Multiplying both sides of ~\eqref{eq:integrate} by $K(t)$ and subtract the resulting equation from ~\eqref{eq:Durhamel} we obtain that
\begin{align}\label{eq:prelim}
[1-K(t)]P_{t}=&Z\int_{0}^{t} [K(t-s)-K(t)] Re \langle W,\ e^{i\frac{\Delta}{2}  s} W\rangle \int_{0}^{s} P_{s_{1}}\ ds_{1} ds\\
 &+2ZK(t)Re \langle W, (i\Delta)^{-1} \int_{0}^{t}  e^{i\frac{\Delta}{2}(t-s)}\  P_{s} \ ds\ W\rangle\nonumber\\
 &+\int_{0}^{t}\ [K(t-s)-K(t)]\ F_{P}(s)\ ds.\nonumber
\end{align}
The first term on the right hand side of ~\eqref{eq:prelim} is rewritten as follows:
\begin{align*}
 &\int_{0}^{t} [K(t-s)-K(t)] Re \langle W,\ e^{i\frac{\Delta}{2}  s} W\rangle \int_{0}^{s} q_{s_{1}}\ ds_{1} ds\\
=&-\int_{0}^{t} [K(t-s)-K(t)] Re\langle W, e^{i\frac{\Delta}{2}  s}W\rangle \int_{s}^{t} q_{s_{1}} ds_{1} ds\\
 &+\int_{0}^{t} K(t-s) Re\langle W, e^{i\frac{\Delta}{2}  s}W\rangle ds\ \int_{0}^{t} q_{s_{1}} ds_{1}\\
 &-K(t) \int_{0}^{t} Re\langle W, e^{i\frac{\Delta}{2}  s}W\rangle ds \ \int_{0}^{t} q_{s_{1}} ds_{1}\\
=&-\int_{0}^{t} [K(t-s)-K(t)] Re\langle W, e^{i\frac{\Delta}{2}  s}W\rangle \int_{s}^{t} q_{s_{1}} ds_{1} ds\\
 &+\int_{0}^{t} K(t-s) Re\langle W, e^{i\frac{\Delta}{2}  s}W\rangle ds\ \int_{0}^{t} q_{s_{1}} ds_{1}\\
 &-2 K(t) Re\langle W, (i\Delta)^{-1} e^{i\frac{\Delta}{2}t} W\rangle  \ \int_{0}^{t} q_{s_{1}} ds_{1}.
\end{align*}
Plugging this expression into the right hand side of ~\eqref{eq:prelim}, we find that
\begin{align}
[1-K(t)]P_{t}=&-Z\int_{0}^{t} [K(t-s)-K(t)] Re\langle W, e^{i\frac{\Delta}{2}  s}W\rangle \int_{s}^{t} P_{s_{1}} ds_{1} ds\nonumber\\
 &+Z\int_{0}^{t} K(t-s) Re\langle W, e^{i\frac{\Delta}{2}  s}W\rangle ds\ \int_{0}^{t} P_{s_{1}} ds_{1}\nonumber\\
 &+2 Z K(t) Re\langle W, (i\Delta)^{-1} \int_{0}^{t}[ e^{i\frac{\Delta}{2}(t-s)} - e^{i\frac{\Delta}{2}t} ] \ P_{s}\ ds\ W\rangle\nonumber\\
 &+\int_{0}^{t}[K(t-s)-K(t)]F_{P}(s)\ ds\nonumber\\
=:&A(P)(t)+\int_{0}^{t}[K(t-s)-K(t)]F_{P}(s)\ ds,\label{eq:finalForm}
\end{align}
thus
\begin{align}\label{eq:secQt}
P_{t}=\frac{1}{1-K(t)} A(P)(t)+\frac{1}{1-K(t)}\int_{0}^{t}[K(t-s)-K(t)]F_{P}(s)\ ds
\end{align}
where $A$ is the linear operator on the space of momentum trajectories $\{P_s\}_{0\leq s<\infty}$ defined by ~\eqref{eq:finalForm}.

Equation ~\eqref{eq:secQt} has the desired form. We will show that, if the trajectories $\{P_{s}\}_{0\leq s<\infty}$ belong to an appropriate Banach space then a fixed-point theorem can be applied that implies global existence of solutions of ~\eqref{eq:secQt}. The main heuristic ideas underlying our approach are discussed in Appendix ~\ref{sub:tacitIdeas}.

Next, we introduce a family of Banach spaces of momentum trajectories: For an arbitrary $\delta\in I,$ where $I$ is the interval defined in ~\eqref{eq:difI}, and any $T>0$, we define the space
\begin{align}\label{eq:defBDeltaT}
B_{\delta, T}:=\{h: [T,\infty)\rightarrow \mathbb{R} \ |t^{\frac{1}{2}+\delta}h\in L^{\infty}[T, \infty)\}
\end{align}
equipped with the norm
\begin{align}\label{eq:norm}
\|h\|_{\delta,T}:= \sup_{t\in [T,\infty)} t^{\frac{1}{2}+\delta}|h|(t) .
\end{align} The function $h(t)$ has the interpretation of being a component of $P_t.$ The definition of $B_{\delta, T}$ can be extended to vector-valued functions $P_t: [T,\infty)\rightarrow \mathbb{R}^3$ in the obvious way.

Below, it will be proven that the operator $A$ in ~\eqref{eq:finalForm} and ~\eqref{eq:secQt} maps the space $B_{\delta,T}$ into itself, for $T$ large enough; (see ~\eqref{eq:estKt}). It appears to be difficult to show that $P_t \in B_{\delta,T}$, for some $\delta>0,$ by starting directly from ~\eqref{eq:Durhamel}. However, in the analysis of ~\eqref{eq:secQt}, a new difficulty appears: Since $K(0)=1,$ the operator  $\frac{1}{1-K(t)} A(\cdot)(t)$ is unbounded on $L^{\infty}[0,\infty)$. The new difficulty is circumvented by ``waiting long enough until $K(t)$ becomes small".
We therefore divide the time axis $[0,\infty)$ into two subintervals, $ [0, T]$ and $ [T,\infty),$ where $T$ is
chosen such that $|K(t)|\ll 1$ when $t\geq T$. For $t\in [0,T]$, a unique strong solution, $P_{t}$, to the equation of motion ~\eqref{eq:linear} corresponding to a given initial condition $P_{t=0}=P_0$ exists, provided $|P_0|$ and $\beta_0$ are small enough, depending on $T$, as shown in Theorem ~\ref{THM:wellposed}. For $t\in [T,\infty)$, we show that a solution $P_t$ exists and belongs to the space $B_{\delta,T}$ by proving that the operator $\frac{1}{1-K(t)}A:\ B_{\delta,T}\rightarrow B_{\delta,T}$ is a contraction on a sufficiently small ball in $B_{\delta,T}$ centered at the origin; evidently this forces us to require that $|P_{T}|$ is small enough, which, by Theorem ~\ref{THM:wellposed}, is guaranteed if $|P_0|$ and $\|\langle x\rangle^{3}\beta_0\|_2$ are chosen to be sufficiently small.

In what follows, we convert this discussion into rigorous mathematics. We define
\begin{equation}\label{eq:cutoff}
\chi_{T}(t):=\left\{
\begin{array}{lll}
1\ \text{if}\ 0\leq t<T\\
0\ \text{if}\ t\geq T
\end{array}
\right.
\end{equation} and rewrite Equation ~\eqref{eq:secQt} as
\begin{equation}\label{eq:secQt2}
P_{t}=\Upsilon(P)(t)+G(t),
\end{equation} where $G(t)$ is the contribution to the right hand side of ~\eqref{eq:secQt} only depending on $\{P_{t}\}_{0\leq t<T}$, i.e., independent of $(1-\chi_{T})P$,
\begin{align}\label{eq:defGt}
G(t):=\frac{1}{1-K(t)}\{ A(\chi_{T}P_{t})+\int_{0}^{t} [K(t-s)-K(t)]F_{\chi_{T} P}(s)\ ds\} .
\end{align}
and $\Upsilon(P)$ contains terms of first order in $(1-\chi_{T})P$ and higher order in $P$ and is given by
\begin{align}
\Upsilon(P)(t)&:=\frac{1}{1-K(t)}\{ A((1-\chi_{T})P)\nonumber\\
&+\int_{0}^{t}
[K(t-s)-K(t)][F_{P}(s)-F_{\chi_{T} P}(s)]\ ds\}.\label{eq:defUpsilon}
\end{align}

In Theorem ~\ref{THM:wellposed}, we have shown that if the initial conditions $P_{0}$ and $\beta_0$ are sufficiently small ($|P_0|,\ \|\langle x\rangle^3 \beta_0\|_2\leq \epsilon_0(T)$) then there exists a unique solution $P_{t},\ t\in [0,T]$, with $\displaystyle\max_{t\in [0,T]}|P_t|\leq T^{-2}.$ In order to continue this solution to the interval $[T,\infty)$ and to show that $\{P_t\}_{T\leq t<\infty}$ belongs to the Banach space $B_{\delta,T}$ we propose to use a fixed-point theorem, which can be applied, provided two conditions are fulfilled:
\begin{itemize}
\item[(1)] The nonlinear map $\Upsilon(\cdot)+G$ maps a small ball, $\mathcal{B}$, in the Banach space $B_{\delta,T}$ centered at $0$ into itself, in particular, $\|G\|_{\delta,T}$ is small enough, and
\item[(2)] $\Upsilon(\cdot)$ is a contraction on $\mathcal{B}$ in the norm $\|\cdot\|_{\delta,T}$ of $B_{\delta,T}$ introduced in ~\eqref{eq:norm}.
\end{itemize}

We begin by verifying that $\Upsilon(\cdot)$ is a contraction on $\mathcal{B}$ if $\mathcal{B}$ is chosen small enough. We define a function $\Omega: [0,1]\rightarrow \mathbb{R}^{+}$ by
\begin{align}\label{eq:difOmega}
\Omega(\delta):=\int_{0}^{1}\frac{1}{1+(1-r)^{\frac{1}{2}}}(1-r)^{-\frac{1}{2}} [\frac{1}{1-2\delta}(r^{-\frac{1}{2}}-r^{-\delta})+ r^{\frac{1}{2}-\delta}]\ dr,
\end{align} (see Equation ~\eqref{eq:difI}). A key result is the following theorem.
\begin{theorem}\label{THM:Contraction}
There exists a constant $T_0<\infty$ such that if $T\geq T_0$ and if $\delta$ is chosen such that $\pi^{-1}\Omega(\delta)<1$ then the map
$\Upsilon(\cdot)$ introduced in ~\eqref{eq:defUpsilon} maps $B_{\delta,T}$ into itself and is a contraction on a sufficiently small ball $\mathcal{B}\subset B_{\delta,T}$ centered at $0.$ The two terms on the right hand side of ~\eqref{eq:defUpsilon} defining $\Upsilon(\cdot)$ satisfy the following estimates:
\begin{itemize}
\item[(1)] The linear operator $A((1-\chi_{T})\cdot)$ satisfies
\begin{equation}\label{eq:contractive}
|\frac{1}{1-K(t)} A((1-\chi_{T})h)(t)|\leq t^{-\frac{1}{2}-\delta}[\frac{1}{\pi}\Omega(\delta)+\epsilon(T)] \|h\|_{\delta,T}
\end{equation} where $\epsilon(T)$ is a small constant satisfying $\displaystyle\lim_{T\rightarrow \infty}\epsilon(T)=0.$
\item[(2)] Let $\epsilon_{0}(T)$ be the constant introduced in Theorem ~\ref{THM:wellposed}. Suppose that $Q_1,\ Q_{2}:[0,\infty)\rightarrow \mathbb{R}^3$ are any two vector-valued functions satisfying the following two conditions: $$Q_1(t)=Q_2(t)=P_{t},\ t\in [0,T],$$
    where $P_t$ is the solution of Equation ~\eqref{eq:linear}/~\eqref{eq:secQt} constructed in Theorem ~\ref{THM:wellposed}, for $0\leq t\leq T$ and $|P_0|,\ \|\langle x\rangle^{3}\beta_0\|_{2}\leq \epsilon_{0}(T)$; and $$\| Q_1\|_{\delta,T}, \ \|Q_2\|_{\delta,T}\ll 1.$$ Then
\begin{align}\label{eq:estFp}
 |\frac{1}{1-K(t)}\int_{0}^{t} [K(t-s)-K(t)][F_{Q_1}-F_{Q_2}](s)\ ds| \lesssim t^{-\frac{1}{2}-\delta}\|Q_1-Q_2\|_{\delta,T}[\|Q_1\|_{\delta,T}+\|Q_2\|_{\delta,T}].
\end{align}
\end{itemize}
\end{theorem}
Inequality ~\eqref{eq:contractive} will be reformulated as Proposition ~\ref{Prop:ThreeTerm}, below, and
proven in Appendix ~\ref{SEC:contraction}. Inequality ~\eqref{eq:estFp} is proven in Appendix ~\ref{sec:highorderterm}.

Next, we present an estimate on the term $G(t)$ on the right hand side of ~\eqref{eq:secQt2} defined in ~\eqref{eq:defGt}. This term only depends on the solution, $P_t$, of ~\eqref{eq:linear}/~\eqref{eq:secQt} for $ t\in [0,T]$, which has been constructed in Theorem ~\ref{THM:wellposed}.
\begin{theorem}\label{THM:smallness}
Suppose that $|P_0|,\ \|\langle x\rangle^{3}\beta_0\|_2\leq \epsilon_0(T)$ and that the parameter $\delta$ is chosen as in Theorem ~\ref{THM:Contraction}. Then $G$ belongs to the Banach space $B_{\delta,T}$ and $\|G\|_{\delta,T}$ can be made arbitrarily small by choosing $T$ large enough; (see Theorem ~\ref{THM:wellposed}).

More specifically, we have that, for any $t\geq T,$
\begin{equation}\label{eq:finite1}
|\frac{1}{1-K(t)}\int_{0}^{t} [K(t-s)-K(t)]F_{\chi_{T} P}(s)\ ds|\leq \epsilon(T) t^{-\frac{1}{2}-\delta},
\end{equation} and
\begin{equation}\label{eq:finite2}
 |\frac{1}{1-K(t)} A(\chi_{T}P_{t})|\leq \epsilon(T) t^{-\frac{1}{2}-\delta},
\end{equation} with $\epsilon(T)\rightarrow 0$ as $T\rightarrow \infty.$
\end{theorem}
\begin{remark}
The constant $\epsilon_0(T)$ is chosen as in Theorem ~\ref{THM:wellposed}.
\end{remark}

The proof of this theorem is contained in Appendix ~\ref{Sec:nonlin}.

In the remainder of this section, we discuss the strategy used to prove Theorem ~\ref{THM:Contraction}. (The proof of Theorem ~\ref{THM:smallness} is easier than the one of Theorem ~\ref{THM:Contraction} and is therefore not discussed here.)
We recall that the map $\Upsilon(\cdot)$ is the sum of two maps appearing on the right hand side of ~\eqref{eq:defUpsilon}. The map
$A((1-\chi_{T})\cdot)(t)$ is linear, while $\frac{1}{1-K(t)}\int_{0}^{t}
[K(t-s)-K(t)][F_{P}(s)-F_{\chi_{T}P}(s)]\ ds$ contains higher-order terms, besides depending on the initial conditions, $\beta_0$, of the Bose gas. Since we are attempting to construct {\bf{small}} solutions of ~\eqref{eq:linear}/~\eqref{eq:secQt}, and because $\beta_0$ can be chosen as small as needed, it is fairly easy to control the second map. We therefore focus our attention on the ideas needed to estimate the first (linear) map. For this purpose, we write the operator $A((1-\chi_{T})\cdot)$ as a sum of three terms:
\begin{align}
A((1-\chi_{T})h)(t)=\sum_{k=1}^{3}\Gamma_{k}(t),\ \Gamma_k(t)\equiv \Gamma_{k}((1-\chi_{T})h)(t),
\end{align}
where
$$
\Gamma_{1}(t):=-Z\int_{0}^{t}\ ds\ [K(t-s)-K(t)] Re\langle W, e^{i\frac{\Delta}{2}  s}W\rangle \int_{s}^{t}\ ds_1\ h_{s_{1}}\ [1- \chi_{T}(s_1)],
$$
$$\Gamma_2(t):=Z\int_{0}^{t}\ ds\ K(t-s) Re\langle W, e^{i\frac{\Delta}{2}  s}W\rangle\ \int_{0}^{t} \ ds_1\ h_{s_{1}} [1-\chi_{T}(s_1)],$$
and
\begin{align}\label{eq:defGamma3}
\Gamma_3(t):=2 Z K(t) Re\langle W, (i\Delta)^{-1} \int_{0}^{t}\ ds\ [ e^{i\frac{\Delta}{2}(t-s)} - e^{i\frac{\Delta}{2}t} ] \ h_{s}\ [1-\chi_{T}(s)]\ W\rangle,
\end{align}
see Eqs. ~\eqref{eq:finalForm} and ~\eqref{eq:cutoff}.

Before estimating $\Gamma_{k}(t),\ k=1,2,3,$ we introduce two functions:
\begin{align}
\Omega_1 (\delta)&:= \frac{1}{(1-2\delta)\pi} \int_{0}^{1}\frac{1}{1+(1-r)^{\frac{1}{2}}} (1-r)^{-\frac{1}{2}} [r^{-\frac{1}{2}} - r^{-\delta}] \ dr,\nonumber\\
\text{and}\ \ \ \ \ \ \label{eq:dif12}\\
\Omega_2(\delta)&:= \frac{1}{\pi}\int_{0}^{1}\frac{1}{1+(1-r)^{\frac{1}{2}}} (1-r)^{-\frac{1}{2}} r^{\frac{1}{2}-\delta} \ dr,\nonumber
\end{align}
see also ~\eqref{eq:difOmega}.

Control of the terms $\Gamma_{k}(t),\ k=1,2,3,$ is provided in the following proposition.
\begin{proposition}\label{Prop:ThreeTerm}  There exists a function $\tilde\epsilon(T), \ 0<T<\infty,$ with $\displaystyle\lim_{T\rightarrow \infty}\tilde\epsilon(T)=0,$ such that, for an arbitrary function $h\in B_{\delta,T},$
\begin{align}\label{eq:Ga1}
|\Gamma_1(t)|\leq t^{-\frac{1}{2}-\delta} [\Omega_1(\delta)+\tilde\epsilon(T)] \|h\|_{\delta,T},
\end{align}
\begin{align}\label{eq:Ga3}
|\Gamma_3(t)|\leq t^{-\frac{1}{2}-\delta} [\Omega_2(\delta)+\tilde\epsilon(T)] \|h\|_{\delta,T},
\end{align}
and
\begin{align}\label{eq:Ga2}
|\Gamma_2(t)|\leq \tilde\epsilon(T) t^{-\frac{1}{2}-\delta} \|h\|_{\delta,T},
\end{align} where $B_{\delta,T}$ and $\|\cdot\|_{\delta,T}$ are as in ~\eqref{eq:defBDeltaT} and ~\eqref{eq:norm} respectively.
\end{proposition}
This proposition will be proven in Appendix ~\ref{SEC:contraction}.
Obviously it implies inequality ~\eqref{eq:contractive} in Theorem ~\ref{THM:Contraction}, $(1),$ with $\epsilon(T)=3\tilde\epsilon(T)$. The ideas underlying the proof of Proposition ~\ref{Prop:ThreeTerm} are as follows. The arguments needed to estimate $\Gamma_1(t)$ and $\Gamma_3(t)$ are very similar, so we only consider $\Gamma_1(t)$ and $\Gamma_2(t).$ The trickiest estimate is ~\eqref{eq:Ga2}. The crucial step is to control the factor
$$I(t):=\int_{0}^{t} K(t-s) ZRe\langle W, e^{i\frac{\Delta}{2}  s}W\rangle\ ds$$ in $\Gamma_2(t).$
We propose to show that
\begin{align}\label{eq:estI}
I(t)=O((1+t)^{-\frac{3}{2}}),\ \text{as}\ t\rightarrow \infty,
\end{align} which implies ~\eqref{eq:Ga2} by straightforward arguments.

Estimate ~\eqref{eq:estI} does not follow by just using that $K(t)=O(t^{-\frac{1}{2}})$ and $\langle W,  e^{i\frac{\Delta}{2}t} W\rangle\ =O( (1+t)^{-\frac{3}{2}})$. These estimates, by themselves, only imply that
$|I(t)|\leq \text{const}\ t^{-\frac{1}{2}}.$ In order to conclude the improved estimate claimed in ~\eqref{eq:estI}, we Fourier-transform the convolution of $K$ with $Re\langle W,  e^{i\frac{\Delta}{2}t} W\rangle \chi_{0}$, with $\chi_0(t)=1,$ for $t\geq 0$, and $=0$, otherwise, which yields
$$I(t)= \frac{Z}{2\pi}\int_{-\infty}^{\infty} \hat{K}(k) \widehat{Re\langle W,  e^{i\frac{\Delta}{2}t} W\rangle \chi_0}(k) e^{-ikt} \ dk.$$
From ~\eqref{eq:DIEkt}, ~\eqref{eq:KtFourier} and ~\eqref{eq:difGKi0} we derive by inspection that
$\widehat{Re\langle W,  e^{i\frac{\Delta}{2}t} W\rangle\chi_0}(k)=-G(k+i0)$, hence
\begin{align}\label{eq:Iform}
I=-\frac{Z}{2\pi} \int_{-\infty}^{\infty} \frac{G(k+i0)}{ik+ZG(k+i0)} e^{-ikt}\ dk,
\end{align} with $G(k+i0)$ as in ~\eqref{eq:difGKi0}. The function $\frac{G(k+i0)}{ik+ZG(k+i0)}$ is smooth in $k$ on $\mathbb{R}\backslash\{0\}$. It is therefore its behavior near $k=0$ that determines the decay of $I(t)$ in $t.$ We recall that $G(k+i0)=ck^{\frac{1}{2}}+O(k),$ for $|k|$ small, where $c$ is some non-zero constant; see ~\eqref{eq:TaylorGk}. Thus
\begin{align}\label{eq:taylorExpa}
-Z\frac{G(k+i0)}{ik+ZG(k+i0)}=-1+\frac{1}{ic}k^{\frac{1}{2}}+O(k),\ \text{for}\ |k| \ \text{small}.
\end{align} Furthermore $ \frac{G(k+i0)}{ik+ZG(k+i0)}$ decays
rapidly in $k$ at infinity. Thus, in Equation ~\eqref{eq:Iform}, we can integrate by parts in the variable $k$, and this yields the desired decay estimate on $I(t)$; (note that the constant term on the right hand side of ~\eqref{eq:taylorExpa} yields a subleading contribution).

Next, we turn to estimating $\Gamma_1$; see ~\eqref{eq:Ga1}. For this purpose, we peel off the main contributions to the functions $K(t)$ and $Re\langle W,  e^{i\frac{\Delta}{2}t} W\rangle:$ By explicit calculation, see ~\eqref{eq:estKt} and ~\eqref{eq:asymp}, one finds that, as $t\rightarrow \infty$, $K(t)=\frac{1}{4Z}\pi^{-\frac{5}{2}} t^{-\frac{1}{2}}+ O(t^{-1});$ see ~\eqref{eq:estKt}, and
$Re\langle W,  e^{i\frac{\Delta}{2}t} W\rangle= C_{W} t^{-\frac{3}{2}}+O((1+t)^{-\frac{5}{2}}).$
We define an approximation, $\tilde{\Gamma}_1(t),$ of $\Gamma_{1}(t)$ by
\begin{align*}
\tilde\Gamma_1:= -\frac{1}{4}\pi^{-\frac{5}{2}} C_{W}\int_{0}^{t}\ ds\  [(t-s)^{-\frac{1}{2}}-t^{-\frac{1}{2}}]
s^{-\frac{3}{2}}\int_{s}^{t}\ ds_1\ h_{s_{1}}\ [1-\chi_{T}(s_1)].
\end{align*}
Recalling the definition of the Banach space $\mathcal{B}_{\delta,T}$ and changing variables, $s=t\sigma$ and $s_1=t\sigma_1,$ we find that
\begin{align*}
|\tilde\Gamma_1|\leq& t^{-\frac{1}{2}-\delta}\frac{1}{4}C_{W}\int_{0}^{1} [(1-\sigma)^{-\frac{1}{2}}-1]
\sigma^{-\frac{3}{2}}\int_{\sigma}^{1} \sigma_1^{-\frac{1}{2}-\delta}\ d\sigma_1 d\sigma\ \|h\|_{\delta,T}\\
\leq &t^{-\frac{1}{2}-\delta} \Omega_1(\delta) \|h\|_{\delta,T}.
\end{align*} To complete our estimate on $\Gamma_1$ we are left with estimating
$(\Gamma_{1}-\tilde{\Gamma}_1)(t)$, which is a straightforward task. In fact $\Gamma_1(t)-\tilde{\Gamma}_1(t)$ decays in $t$ faster than $\tilde\Gamma_1(t).$

Further details of our estimates needed to prove Proposition ~\ref{Prop:ThreeTerm} can be found in Subsection ~\ref{subsec:Ga1} of Appendix ~\ref{SEC:contraction}.

\section{Proof of the Main Result, Theorem ~\ref{THM:main}}\label{sec:ProofMainTheorem}
We use the subdivision of the time axis $$[0,\infty)=[0,T)\cup [T,\infty)$$ into two parts, for an appropriately chosen $T.$ Existence and uniqueness of a solution $P_t,$ of Equations ~\eqref{eq:linear}/~\eqref{eq:secQt}, for $t\in [0,T),$ assuming that the initial conditions $P_0$ and $\beta_0$ are small enough (depending on $T$), has been proven in Theorem ~\ref{THM:wellposed}.

To continue such a solution to the interval $[T,\infty),$ we apply a standard fixed-point theorem to Equation ~\eqref{eq:secQt2}. Thanks to Theorems ~\ref{THM:Contraction} and ~\ref{THM:smallness}, the hypotheses of the fixed-point theorem are valid, provided $T$ is chosen appropriately, and $|P_0|$, $\|\langle x\rangle^3\beta_0\|_2$ are small enough. We thus conclude that a global solution $P_t,\ t\in [0,\infty),$ to ~\eqref{eq:linear}/~\eqref{eq:secQt} exists, with $(1-\chi_{T})P_t\in B_{\delta,T},$ for any $\delta$ in the interval $I$ defined in ~\eqref{eq:difI}, provided $|P_0|$ and $\|\langle x\rangle^3 \beta_0\|_2$ are chosen small enough. This proves estimate ~\eqref{eq:trajectory} of Theorem ~\ref{THM:main}.

In order to prove Eq.~\eqref{eq:convergence} in Theorem ~\ref{THM:main}, we show that the field $\delta_t$ introduced in ~\eqref{eq:decom} decays to $0$, as $t\rightarrow \infty,$ in the sense that
\begin{equation}\label{eq:equiva}
\|\delta_{t}\|_{\infty}\rightarrow 0,\ \text{as}\ t\rightarrow \infty.
\end{equation}

To establish ~\eqref{eq:equiva} we apply the norm $\|\cdot\|_{\infty}$ to both sides of ~\eqref{eq:Duh}, which yields
$$
\|\delta_{t}\|_{\infty}\leq \text{const.}\{ \| e^{i\frac{\Delta}{2}t} (-\Delta)^{-1}W^{X_{0}}\|_{\infty}+\|e^{i\frac{\Delta t}{2} }\beta_{0}\|_{\infty}+\int_{0}^{t} \| e^{i\frac{\Delta}{2}t} (-\Delta)^{-1} P_{s} \nabla_{x} W^{X_{s}}\|_{\infty}\ ds\}.
$$ Using the estimates $$\| e^{i\frac{\Delta}{2}t} (-\Delta)^{-1}\|_{L^{1}\rightarrow L^{\infty}}\leq \text{const}\ t^{-\frac{1}{2}}$$ and $$\| e^{i\frac{\Delta}{2}t} \|_{L^{1}\rightarrow L^{\infty}}\leq \text{const}\ t^{-\frac{3}{2}}$$ and our estimate on $P_t$, see ~\eqref{eq:trajectory}, we find that
\begin{align*}
\|\delta_{t}\|_{\infty}\leq \text{const}\ t^{-\frac{1}{2}} +t^{-\frac{3}{2}}\|\beta_{0}\|_{L^{1}}
\end{align*} which, under our assumption on $\beta_0,$ yields ~\eqref{eq:equiva}.

This completes the proof of Theorem ~\ref{THM:main}. Some hard technicalities now follow in several appendices.
\begin{flushright}
$\square$
\end{flushright}


\appendix


\section{Proof of Proposition ~\ref{Prop:kernel}}\label{Sec:ReforPro}
We start with deriving an explicit formula for $K,$ see ~\eqref{eq:DIEkt}.

We define a function $G:\mathbb{R}\rightarrow \mathbb{C}$ by
\begin{equation}\label{eq:difGk}
G(k+i0):=-i\langle (-\Delta+2k+i0)^{-1}W, W\rangle +i \langle (-\Delta-2k-i0)^{-1}W, W\rangle.
\end{equation}
Next, we relate $G$ to the function $K$.
\begin{proposition}\label{Prop:estK}
The function $K$ in ~\eqref{eq:DIEkt} takes the form
\begin{equation}\label{eq:ReQt}
K(t)
=-(2\pi)^{-1}\int_{-\infty}^{\infty} \frac{1}{ik+Z G(k+i0)} e^{-ikt}\ dk
\end{equation} in particular
\begin{equation}\label{eq:zero}
K(t)=0\ \text{for}\ t<0.
\end{equation}
The function $K$ can be transformed to a convenient form
\begin{equation}\label{eq:convenient}
K(t)=-\frac{1}{\pi}\int_{-\infty}^{\infty} Re \frac{1}{ik+Z G(k+i0)}\ coskt\ dk.
\end{equation}
\end{proposition}
This proposition is proven in Subsections ~\ref{subsec:a1} and ~\ref{subsec:a2}. The basic ideas in the proof are not difficult. In a formal level Eq.~\eqref{eq:ReQt} is obtained by Fourier transformations, as mentioned after Proposition ~\ref{Prop:kernel}. In Subsection ~\ref{subsec:a2} we will make it rigorous. Eq.~\eqref{eq:zero} is resulted by the facts that $\frac{1}{iz+Z G(z)}$ is analytic in the set $Im\ z>0$ and its absolute value is sufficiently small when $|z|$ is large. Hence $$K(t)
=-(2\pi)^{-1}\int_{-\infty}^{\infty} \frac{1}{i(k+ia)+Z G(k+a)} e^{-i(k+ia)t}\ dk$$ for any $a>0,$ and moreover $K(t)\rightarrow 0$ as $a\rightarrow \infty$ if $t<0.$ We obtain the last identity ~\eqref{eq:convenient} by manipulating the expression in ~\eqref{eq:ReQt}.

To prove Proposition ~\ref{Prop:kernel} it suffices to derive a decay estimate for $K(t)$ from ~\eqref{eq:convenient} using the oscillatory nature of $cos kt$.
Since the function $Re\frac{1}{ik+ZG(k+i0)}:\mathbb{R}\rightarrow \mathbb{R}$ is smooth on the open set $(-\infty,\infty)\backslash\{0\}$, it is the lowest order term in the Taylor-expansion of the function in a neighborhood of $k=0$ that determines the decay in $t$.
\begin{lemma}\label{LM:TaylorExp}
The function $G(k+i0)$ defined in ~\eqref{eq:difGk} satisfies the estimate
\begin{equation}\label{eq:Gki0}
G(k+i0)=\left\{
\begin{array}{lll}
2^{\frac{3}{2}}(i-1)\pi^2 k^{\frac{1}{2}}+ C k+O(|k|^{\frac{3}{2}})\ \ \ \text{if}\ k>0\\
2^{\frac{3}{2}}(-i-1)\pi^2 |k|^{\frac{1}{2}}+C k+O(|k|^{\frac{3}{2}})\ \text{if}\ k<0\\
\end{array}
\right.
\end{equation} where $C$ is some constant.
\end{lemma}
This lemma is proven in Subsection ~\ref{subsec:Taylor}, by Taylor-expanding the function $G(k+i0)$ in variable $k^{\frac{1}{2}}.$

Now we are ready to prove Proposition ~\ref{Prop:kernel}\\
{\bf{Proof of Proposition ~\ref{Prop:kernel}}}
To simplify matters we decompose $K(t)$ into two parts, according to the integration regions
\begin{equation}\label{eq:decomK}
K(t)=K_{+}(t)+K_{-}(t),
\end{equation} with $$K_{+}:=-\frac{1}{\pi} \int_{0}^{\infty} Re\frac{1}{ik+ZG(k+i0)} cos kt\ dk$$ and $$K_{-}:=-\frac{1}{\pi} \int_{-\infty}^{0} Re \frac{1}{ik+ZG(k+i0)} coskt\ dk.$$
We first estimate $K_{+}$. Since the leading order is determined by Taylor expansion of the integrand $Re\frac{1}{ik+ZG(k+i0)} cos kt$ around $k=0;$ it is natural to begin with studying this function. We define a new function $g:\mathbb{R}^{+}\rightarrow \mathbb{R}$ by
$$|k|^{-\frac{1}{2}} g(|k|^{\frac{1}{2}}):=Re -\frac{1}{\pi} \frac{1}{ik+ZG(k+i0)}.$$
By direct computation, using the result in ~\eqref{eq:Gki0}, we find that, in a neighborhood of $k=0,$
$$
\begin{array}{lll}
|k|^{-\frac{1}{2}} g(|k|^{\frac{1}{2}})
&=&-\frac{1}{\pi Z}\frac{ReG}{(\frac{k}{Z}+Im G)^2+(ReG)^2}
=\frac{1}{2^{\frac{5}{2}}\pi^3 Z}|k|^{-\frac{1}{2}}[1+O(k^{\frac{1}{2}})],
\end{array}
$$ where, in the last step, the result in ~\eqref{eq:Gki0} was used. The other important observations are that the function $g: \mathbb{R}^{+}\rightarrow \mathbb{C}$ is smooth on $[0,\infty)$ and satisfies the estimate $$|g(\rho)|\leq C(1+\rho)^{-3}.$$

Expanding $g(k)$ around $k=0$ we obtain
\begin{align*}
K_{+}(t)=&\int_{0}^{\infty} |k|^{-\frac{1}{2}} g(|k|^{\frac{1}{2}})\ coskt \ dk\\
=& 2\int_{0}^{\infty} g(\rho) cos(\rho^2 t)\ d\rho\\
=& 2 g(0)\int_{0}^{\infty}cos(\rho^2 t)\ d\rho+ D
\end{align*}
where $D$ is given by $$D:=2\int_{0}^{\infty} [g(\rho)-g(0)] cos(\rho^2 t)\ d\rho .$$

The first term on the right hand side is the dominant one,
\begin{equation}\label{eq:appFres}
2 g(0)\int_{0}^{\infty}cos(\rho^2 t)\ d\rho= 2g(0) t^{-\frac{1}{2}}\int_{0}^{\infty} cos x^2\ dx=
\frac{1}{8}\pi^{-\frac{5}{2}} Z^{-1} t^{-\frac{1}{2}},
\end{equation} using the Fresnel integral ~\eqref{eq:Fresnel}.

The second term, $D,$ is of the form
$$
D=\int_{0}^{\infty} H(\rho) cos(\rho^2 t)\rho \ d\rho
$$ where the function $H=\frac{2(g(\rho)-g(0))}{\rho}$ is smooth and is bounded uniformly by $C(1+\rho)^{-1}$. By standard techniques we find that
\begin{equation}\label{eq:difD}
D=O(t^{-\frac{3}{2}}).
\end{equation} This together with ~\eqref{eq:appFres} implies that
\begin{equation}
K_{+}=\frac{1}{8} \pi^{-\frac{5}{2}} Z^{-1} t^{-\frac{1}{2}} +O(t^{-\frac{3}{2}}).
\end{equation}

For $K_{-}$, we obtain, using almost identical arguments,
$$K_{-}=\frac{1}{8} \pi^{-\frac{5}{2}} Z^{-1} t^{-\frac{1}{2}} +O(t^{-\frac{3}{2}}).$$
These two results, together with ~\eqref{eq:decomK}, obviously imply Proposition ~\ref{Prop:kernel}.
\begin{flushright}
$\square$
\end{flushright}

In next three subsections we prove Lemma ~\ref{LM:TaylorExp} and Proposition ~\ref{Prop:estK}.

In what follows we often use Fourier transform. Its definition and properties are standard. Since constants are important in the present paper, we quote some of them explicitly.
For a function $f:\mathbb{R}^{d}\rightarrow \mathbb{C}$, its Fourier transformation $\hat{f}$ is defined as
$$\hat{f}(k):=(2\pi)^{-\frac{d}{2}} \int_{\mathbb{R}^{d}} e^{ik\cdot x} f(x)\ dx,$$ and the inverse transform as $$\check{f}(x):=(2\pi)^{-\frac{d}{2}} \int_{\mathbb{R}^{d}} e^{-ik\cdot x} f(k)\ dk.$$ Moreover, for arbitrary functions $f,\ g:\ \mathbb{R}^{d}\rightarrow \mathbb{C}$
\begin{equation}\label{eq:FTconv}
\widehat{fg}=(2\pi)^{-\frac{d}{2}}\int_{\mathbb{R}^d} \hat{f}(x-y)\hat{g}(y)dy.
\end{equation}
\subsection{Proof of Lemma ~\ref{LM:TaylorExp}}\label{subsec:Taylor}
\begin{proof}
By Fourier-transformation and introducing polar coordinates we find that $G(k+i0)$ takes the convenient form
\begin{align}
G(k+i0)=&-i\langle (\rho^2+2k+i0)^{-1}\hat{W}, \hat{W}\rangle +i\langle (\rho^2-2k-i0)^{-1}\hat{W}, \hat{W}\rangle\nonumber\\
=&i8\pi k [\int_{0}^{\infty} (\rho^2+2k+i0)^{-1}|\hat{W}(\rho)|^2 d\rho+\int_{0}^{\infty} (\rho^2-2k-i0)^{-1}|\hat{W}(\rho)|^2 d\rho].\label{eq:GKFT}
\end{align}

We expand $|\hat{W}(\rho)|^2$ in a neighborhood of $\rho=0.$ The fact that the function $W:\mathbb{R}^3\rightarrow \mathbb{R}$ is smooth, spherically symmetric and decays rapidly at $\infty$ implies that the function $\hat{W}$ is smooth in the variable $\rho^2, $ or, equivalently, that there exists a smooth and rapidly decaying function $F:\mathbb{R}^{+}\rightarrow \mathbb{R}$ such that
\begin{equation}\label{eq:Tal}
|\hat{W}(\rho)|^2=1+\rho^2 F(\rho^2).
\end{equation}
Here the condition that $|\hat{W}(0)|=1,$ in ~\eqref{eq:scale}, is used. Plugging ~\eqref{eq:Tal} into ~\eqref{eq:GKFT} we find that
$$
G(k+i0)=G_0(k)+G_1(k)+G_2(k)
$$
with $$G_0(k+i0):=i 8\pi k [\int_{0}^{\infty} (\rho^2+2k+i0)^{-1} d\rho+\int_{0}^{\infty} (\rho^2-2k-i0)^{-1} d\rho],$$
$$G_1(k+i0):=8i\pi k \int_{0}^{\infty} F(\rho^2)\ d\rho,$$ and
$$G_2(k+i0):=-16i\pi k^2 [\int_{0}^{\infty} (\rho^2+2k+i0)^{-1}F(\rho^2)\ d\rho-\int_{0}^{\infty} (\rho^2-2k-i0)^{-1}F(\rho^2)\ d\rho].$$

In the next we estimate the three terms.

The estimate on $G_1$ is evident:
\begin{equation}
G_{1}(k+i0)=O(k).
\end{equation}

The function $G_0(k+i0)$ has an explicit expression: For any $k\in \mathbb{C}\backslash \mathbb{R}^{-},$ we observe that
\begin{equation}
\int_{0}^{\infty} (\rho^2+2k)^{-1} d\rho=\frac{\pi}{2\sqrt{2}} k^{-\frac{1}{2}},
\end{equation} where $k^{-\frac{1}{2}}=|k|^{-\frac{1}{2}},$ for $k\in \mathbb{R}^{+}.$
The proof consists in observing that
if $k>0$ then $$\int_{0}^{\infty} (\rho^2+2k)^{-1} d\rho=\frac{1}{\sqrt{2}} k^{-\frac{1}{2}} \int_{0}^{\infty} (\rho^2+1)^{-1} d\rho=
\frac{\pi}{2\sqrt{2}} k^{-\frac{1}{2}}.$$ Consequently
\begin{equation}
G_0(k+i0)=\left\{
\begin{array}{lll}
2^{\frac{3}{2}}(i-1)\pi^2 k^{\frac{1}{2}}\ \ \ \ \ \  \text{if}\ k>0\\
2^{\frac{3}{2}}(-i-1)\pi^2 |k|^{\frac{1}{2}}\ \ \ \text{if}\ k<0\\
\end{array}
\right.
\end{equation}

To estimate $G_2$ it is sufficient to show that if a function $\phi$ decays sufficiently fast at $\infty$ then, for any small $k$,
\begin{equation}\label{eq:estG2}
\int_{-\infty}^{\infty} (\rho^2\pm 2k\pm i0)^{-1}\phi(\rho)\ d\rho=O(|k|^{-\frac{1}{2}}).
\end{equation} Indeed, we use Fourier transformation to relate $(\rho^2\pm 2k\pm i0)^{-1}$ to $(-\partial_{x}^2\pm 2k\pm i0)^{-1}$
and find that
$$(-\partial_{x}^2 \pm 2k\pm i0)^{-1} \hat{\phi}|_{x=0}= C \int_{-\infty}^{\infty} (\rho^2\pm 2k\pm i0)^{-1}\phi(\rho)\ d\rho$$
for some constant $C\not=0.$ The operator $(-\partial_{x}^2 + 2k),\ k\in \mathbb{C}\backslash{R}^{-},$
has an integral kernel $C_1 k^{-\frac{1}{2}} e^{-k^{\frac{1}{2}}|x-y|},$ where $C_1$ is a constant. This yields $$|\int_{-\infty}^{\infty} (\rho^2\pm 2k\pm i0)^{-1}\phi(\rho)\ d\rho|\lesssim |k|^{-\frac{1}{2}}\|\phi\|_{L^{1}}$$
which is the desired estimate ~\eqref{eq:estG2}.

Collecting the estimates above we complete our proof.
\end{proof}

\subsection{Proof of ~\eqref{eq:zero} of Proposition ~\ref{Prop:estK}}\label{subsec:a1}
We start by extending the domain of function $G$ from $k\in \mathbb{R}$ to $Im\ k>0$.
By Fourier transformation we find that $G(k+i0)$ takes the form
\begin{equation}
G(k+i0)=-i\langle (\rho^2+2k+i0)^{-1}\hat{W}, \hat{W}\rangle +i\langle (\rho^2-2k-i0)^{-1}\hat{W}, \hat{W}\rangle.
\end{equation}
Its extension is the function $G:\ \{k| Imk>0\}\rightarrow \mathbb{C}$ defined by
$$G(k):=8i\pi k [\int_{0}^{\infty} (\rho^2+2k)^{-1}|\hat{W}(\rho)|^2 d\rho+\int_{0}^{\infty} (\rho^2-2k)^{-1}|\hat{W}(\rho)|^2 d\rho].$$ It is easy to see that this function is analytic in $k,\ Imk>0$, and $G(k+i0),\ k\in \mathbb{R},$ is its limit on the real line.

~\eqref{eq:zero} follows by contour integration. To guarantee its applicability, we have to verify several criteria.

We start with the following result.
\begin{lemma}\label{LM:analytic} In the complex region $Imk>0,$ the function
$\frac{1}{ik+Z G(k)} e^{ik t} $ is analytic in $k$ for any (fixed) $t,$ and
\begin{equation}\label{eq:anaSma}
\displaystyle\lim_{|k|\rightarrow \infty}|\frac{1}{ik+Z G(k)}|\rightarrow 0.
\end{equation}
\end{lemma}
\begin{proof}
~\eqref{eq:anaSma} is implied by the fact that $G(k)\rightarrow 0,$ as $|k|\rightarrow \infty.$

It is easy to see that the function $ik+Z G(k)$ is analytic, in the region $Im k>0,$ because
the operator $(-\Delta\pm 2k)^{-1}$ is well defined and analytic in $k$. To prove analyticy of $e^{ikt}\frac{1}{ik+ZG(k)}$ we only need to prove that the denominator does vanish anywhere, i.e.
\begin{equation}\label{eq:Fnonzero}
|ik+Z G(k)|\not=0 \ \text{when}\ Im k>0.
\end{equation}
For this purpose we rewrite the expression of $G(k)$ in ~\eqref{eq:ForTran} to obtain
\begin{equation}\label{eq:newForm}
ik+ZG(k)= ik[1+ 8Z\pi \int_{0}^{\infty} \rho^2(\rho^{4}-4k^2) |\hat{W}(\rho)|^2\ d\rho].
\end{equation}
In what follows we consider two cases, $Re k=0$ and $Rek\not=0$ separately.
\begin{itemize}
\item[(A)] If $Re k=0$ then, by the fact that $Im k>0,$ we find that $- k^2>0$, hence $\rho^4- 4k^2>0,$ and this implies ~\eqref{eq:Fnonzero}.
\item[(B)] If $Rek\not=0$ the key observation is that $\int_{0}^{\infty} \rho^2(\rho^{4}-4k^2) |\hat{W}(\rho)|^2\ d\rho$
has a non vanishing imaginary part. Indeed $- 4k^2$ can be written in the form $- 4k^2=a+ib, $ with $b\not=0.$
We rewrite $ (\rho^{4}-4k^2)^{-1}$ as $ (\rho^4-4k^2)^{-1}=[(\rho^4+a)^2+b^2]^{-1}(\rho^4+a-ib).$
Hence $$Im\int_{0}^{\infty} \rho^2(\rho^{4}-4k^2) |\hat{W}(\rho)|^2\ d\rho=-b\int_{0}^{\infty}[(\rho^4+a)^2+b^2]^{-1} |\hat{W}(\rho)|^2\ d\rho\not=0.$$ This together with ~\eqref{eq:newForm} implies ~\eqref{eq:Fnonzero}.
\end{itemize}
\end{proof}
We continue to prove ~\eqref{eq:zero}. The fact that $\frac{1}{ik+Z G(k)}e^{ikt}$ is analytic on the domain $Im k>0$ and the decay estimate in ~\eqref{eq:anaSma}, proven in Lemma ~\ref{LM:analytic}, imply that, for any $a>0,$
\begin{align*}
F(t)&:=\int_{-\infty}^{\infty}\frac{1}{ik+Z G(k+i0)} e^{-ikt} dk\\
&=\int_{-\infty}^{\infty}\frac{1}{i(k+ia)+Z G(k+ia)} e^{-i(k+ia)t} dk\\
&=e^{a t }\int_{-\infty}^{\infty}\frac{1}{i(k+ia)+Z G(k+ia)} e^{-ik t} dk.
\end{align*}
If $t<0$ it is easy to see that $F(t)=0$ by letting $a\rightarrow +\infty.$ This is ~\eqref{eq:zero}.

We turn to ~\eqref{eq:convenient}. The definition of $G(k+i0)$ implies that $G(k+i0)=-G(-k-i0)$. Moreover, by changing variables $k\rightarrow -k$, we find that, for any $t>0$
$$F(t)=\int_{-\infty}^{\infty}\frac{1}{-ik-ZG(k-i0)} e^{ikt} dk$$ and
$$F(-t)=\int_{-\infty}^{\infty}\frac{1}{-ik-ZG(k-i0)} e^{-ikt} dk=0.$$
These identities, together with ~\eqref{eq:zero} and the observation that $\overline{G(k+i0)}=-G(k-i0),$ imply $$
\begin{array}{lll}
F(t)&=&\frac{1}{2}\int_{-\infty}^{\infty}[\frac{1}{ik+Z G(k+i0)} +\frac{1}{-ik-Z G(k-i0)}][e^{ikt}+e^{-ikt}]\ dk\\
& &\\
&=&2\int_{-\infty}^{\infty}Re \frac{1}{ik+Z G(k+i0)}  coskt \ dk.
\end{array}
$$ The desired result follows by using the definition of $F(t)$.
\begin{flushright}
$\square$
\end{flushright}

\subsection{Proof of ~\eqref{eq:ReQt} of Proposition ~\ref{Prop:estK}}\label{subsec:a2}
The idea underlying the proof is simple, at least on a formal level: By Fourier transforming both sides of ~\eqref{eq:DIEkt} one obtains an explicit expression for $\int_{0}^{\infty}\ e^{ikt} q_{t}\ dt.$ Then, by inverse Fourier transformation, one arrives at the desired result. The setback is that the Fourier transformation and its inverse are defined on $L^{2},$ and
the fact that the function $K\not\in L^2$ makes things a little harder.
In what follows, we begin with another result and show that it implies ~\eqref{eq:DIEkt}.

The result is
\begin{lemma}\label{LM:equiva}
Suppose the function $\tilde{q}_{t}:\mathbb{R}^{+}\rightarrow \mathbb{R}$ satisfies the estimate $|\tilde{q}_{t}|\leq C_1(1+t)^{-2}$ and is the solution to the equation
\begin{equation}\label{eq:linear2}
\partial_{t}\tilde{q}_{t}=- Z Re \langle W,\ \int_{0}^{t}  e^{i\frac{\Delta}{2}(t-s)} \tilde{q}_{s}\ ds\ W\rangle+h(t)
\end{equation} where the function $h:\mathbb{R}^{+}\rightarrow \mathbb{R}$ decays like $|h(t)|\leq c_2 (1+t)^{-\frac{3}{2}}.$ Then $\tilde{q}_{t}$ also satisfies the equation
\begin{equation}
\tilde{q}_{t}= K(t)\tilde{q}_0+\int_{0}^{t} K(t-s) h(s)\ ds.
\end{equation}
\end{lemma}
This lemma will be proven shortly.

We now prove~\eqref{eq:ReQt} by assuming Lemma ~\ref{LM:equiva}.\\

{\bf{Proof of ~\eqref{eq:ReQt}}}
To make Lemma ~\ref{LM:equiva} applicable we define a new smooth function $u:\mathbb{R}^{+}\rightarrow \mathbb{R}$ with the properties $$u_{t}=K(t)\ \text{if}\ t\leq T,\ \text{and}\ |u_{t}|\leq C(1+t)^{-2}$$ for some constant $C.$
This function is well defined, because, in any finite time interval $[0,T],$ the solution $K$ to the linear equation ~\eqref{eq:DIEkt} exists.
By direct computation we find that $u$ satisfies the equation
\begin{equation}\label{eq:ut}
\partial_{t} u_t=- Z Re \langle W,\ \int_{0}^{t}  e^{i\frac{\Delta}{2}(t-s)} u_{s}\ ds\ W\rangle- f_{t}
\end{equation} where the function $f$ is defined by $$f_t:=- Z Re \langle W,\ \int_{0}^{t}  e^{i\frac{\Delta}{2}(t-s)} u_{s}\ ds\ W\rangle-\partial_{t}u_{t}:\mathbb{R}^{+}\rightarrow \mathbb{R}$$ and satisfies the estimates $$f_{t}=0\ \text{if}\ t\leq T,\ \text{and}\ |f_{t}|\leq C_{u}(1+t)^{-\frac{3}{2}}.$$ These estimates are obtained easily, hence we omit the details.

Applying Lemma ~\ref{LM:equiva} to~\eqref{eq:ut} and using the fact that $f_{t}=0$ if $t\leq T,$ we find that if $t\leq T$ then $u_t=K(t)$ takes the desired form. Since $T$ is arbitrary, this equation holds for any time.
\begin{flushright}
$\square$
\end{flushright}

{\bf{Proof of Lemma ~\ref{LM:equiva}}}
We Fourier-transform both sides of ~\eqref{eq:linear2} to derive an expression for $\int_{0}^{\infty} e^{ikt} \tilde{q}_{t} dt$:
$$
\begin{array}{lll}
\int_{0}^{\infty} e^{ikt} \partial_{t}\tilde{q}_{t} dt&=&-\tilde{q}_{0}-ik \int_{0}^{\infty} e^{ikt} \tilde{q}_{t} dt\\
& &\\
&=&-Z \int_{0}^{\infty} e^{ikt} Re \langle W,\ \int_{0}^{t}  e^{i\frac{\Delta}{2}(t-s)}  \tilde{q}_{s}ds\ W\rangle \ dt+ \int_{0}^{\infty} e^{ikt}h(t)\ dt.
\end{array}
$$
An important observation is that the first term on the right hand side admits a simpler expression:
\begin{equation}\label{eq:ForTran}
-\int_{0}^{\infty} e^{ikt} Re \langle W,\ \int_{0}^{t}  e^{i\frac{\Delta}{2}(t-s)}  \tilde{q}_{s}ds\ W\rangle dt= G(k+i0)\int_{0}^{\infty} e^{ikt} \tilde{q}_t\ dt,
\end{equation} where the function $G(\cdot+i0):\mathbb{R}\rightarrow \mathbb{C}$ is defined in \eqref{eq:difGk}.

Collecting the identities above we find that
$$
\int_{0}^{\infty} e^{ikt} \tilde{q}_{t}dt=-\frac{1}{ik+Z G(k+i0)} q_{0}-\frac{1}{ik+Z G(k+i0)}\int_{0}^{\infty} e^{ikt} h(t)\ dt.
$$
Inverse-Fourier transform of both sides of the equation yields
$$
\tilde{q}_{t}= -(2\pi)^{-1}\int_{-\infty}^{\infty} \frac{1}{ik+Z G(k+i0)} e^{-ikt}\ dk\ \tilde{q}_{0}+\frac{1}{2\pi} \int_{-\infty}^{\infty}e^{-ikt} H(k)\ dk,
$$
where $H(k)$ is defined as
$$H(k):=-\frac{1}{ik+Z G(k+i0)}\int_{0}^{\infty} e^{ikt}h(t)\ dt.$$
By applying Fourier transformation, using ~\eqref{eq:FTconv} and the fact that $K(t)=0,$ for $t<0,$ we obtain that
\begin{equation}\label{eq:kernel}
\frac{1}{2\pi} \int_{-\infty}^{\infty}e^{-ikt} H(k)\ dk=\int_{0}^{t}K(t-s)h(s) \ ds.
\end{equation}

Collecting the estimates above we complete the proof.
\begin{flushright}
$\square$
\end{flushright}
\subsection{Some Decay Estimates}
In this part, we collect some estimates used in the main part of the paper.
\begin{lemma}\label{LM:propaW}
If $W$ is a smooth, spherically symmetric and rapidly decaying function satisfying the condition
 $|\hat{W}(0)|=1$ then, as $t\rightarrow \infty$,
\begin{equation}\label{eq:asymp}
\begin{array}{lll}
Re\langle W,  e^{i\frac{\Delta}{2}t} W\rangle&=&-2 t^{-\frac{3}{2}} \pi^{\frac{3}{2}}+O(t^{-\frac{5}{2}})\\
& &\\
2Re\langle W, (i\Delta)^{-1} e^{i\frac{\Delta}{2}t} W\rangle&=& 4 t^{-\frac{1}{2}} \pi^{\frac{3}{2}} +O(t^{-\frac{3}{2}}).
\end{array}
\end{equation}
\end{lemma}
\begin{proof}
We start with proving the first equation. By Fourier transformation, introducing polar coordinates, and integrating by parts we find that
\begin{align}
Re\langle W,  e^{i\frac{\Delta}{2}t} W\rangle &=Re\langle \hat{W}, e^{-i\frac{|k|^2}{2} t}\hat{W}\rangle\nonumber\\
&= 4\pi \int_{0}^{\infty} \rho^2 cos(\frac{1}{2}\rho^2 t) |\hat{W}(\rho)|^2 d\rho\nonumber\\
&= -\frac{4\pi}{t} \int_{0}^{\infty} sin(\frac{1}{2}\rho^2 t) |\hat{W}(\rho)|^2 d\rho-\frac{4\pi}{t}\int_{0}^{\infty} sin(\frac{1}{2}\rho^2 t) \rho \partial_{\rho}[|\hat{W}(\rho)|^2] d\rho\label{eq:computation}\\
&=-\frac{4\pi}{t} [\int_{0}^{\infty} sin(\frac{1}{2}\rho^2 t) d\rho+\int_{0}^{\infty} sin(\frac{1}{2}\rho^2 t) F(\rho)\ d\rho]\nonumber
\end{align} where $F(\rho)$ is defined by $$F(\rho):=|\hat{W}(\rho)|^2 -1+\rho \partial_{\rho}|\hat{W}(\rho)|^2.$$
By rescaling $\rho (\frac{t}{2})^{\frac{1}{2}}\rightarrow \rho$ it is easy to see that the first term on the right hand side takes the form
\begin{equation}\label{eq:sinRho}
-\frac{4\pi}{t} \int_{0}^{\infty} sin(\frac{1}{2}\rho^2 t)\ d\rho=-4\sqrt{2}\pi t^{-\frac{3}{2}}\int_{0}^{\infty} sin(\rho^2)\ d\rho=
-2\pi^{\frac{3}{2}} t^{-\frac{3}{2}}.
\end{equation} Here the fact that $\int_{0}^{\infty} sin x^2 dx=(\frac{\pi}{8})^{\frac{1}{2}}$ is used; (see ~\eqref{eq:Fresnel}).

Now we turn to the second term. Define a new function, $\tilde{F}(\rho),$ by $\tilde{F}(\rho):=\rho^{-2}F(\rho).$ This changes our expression to $$
\int_{0}^{\infty} sin(\frac{\rho^2 t}{2}) F(\rho)\ d\rho =\int_{0}^{\infty}\rho^{2} sin(\frac{\rho^{2}t}{2}) \tilde{F}(\rho)\ d\rho.$$ Thanks to the fact that $F(\rho)=O(\rho^2),$ in a neighborhood of $\rho=0$ and because of smoothness, we obtain that $\tilde{F}$ is a smooth function. By standard techniques it is seen that
\begin{equation}
\int_{0}^{\infty} sin(\frac{\rho^2 t}{2}) F(\rho)\ d\rho=O(t^{-\frac{3}{2}}).
\end{equation}
This together with ~\eqref{eq:sinRho} implies the first estimate in ~\eqref{eq:asymp}.

By similar arguments we obtain the second estimate in ~\eqref{eq:asymp}.
\end{proof}
The following identity has been used in the proof.
\begin{equation}\label{eq:Fresnel}
\int_{0}^{\infty} cos (x^2)\ dx=\int_{0}^{\infty} sin (x^2)\ dx=(\frac{\pi}{8})^{\frac{1}{2}}.
\end{equation}
\begin{proof}
The key observation is that the function $e^{i z^2}=cos z^2+i sin z^2:\ \mathbb{C}\rightarrow \mathbb{C}$ is an entire function. This together with the fact that in the region $Re z, \ Im z>0$ the function decays rapidly, at $|z|\rightarrow\infty,$ enables us to use the method of contour integration to obtain $$\int_{0}^{\infty} e^{i z^2} \ dz=\int_{z\in \{ \frac{1}{\sqrt{2}}x+i\frac{1}{\sqrt{2}}x|\ x\in [0,\infty)\}} e^{i z^2 }\ dz=(\frac{1}{\sqrt{2}}+\frac{1}{\sqrt{2}} i )\int_{0}^{\infty} e^{- z^2}\ dz=\frac{\sqrt{\pi}}{\sqrt{8}}+\frac{\sqrt{\pi}}{\sqrt{8}}i .$$ This obviously implies ~\eqref{eq:Fresnel}.
\end{proof}
\section{Proof of Proposition ~\ref{Prop:ThreeTerm}}\label{SEC:contraction}
In the following we prove the three estimates in Proposition ~\ref{Prop:ThreeTerm}, in three different subsections.
In various places we have to consider separately two cases: $\delta<\frac{1}{2}$, where functions in the Banach spaces
$B_{\delta,T}$, see ~\eqref{eq:defBDeltaT}, are not necessarily integrable in time on the interval $[T,\infty)$,
and $\delta>\frac{1}{2}$, where functions in $B_{\delta,T}$ are integrable.
In order not to clutter our arguments with too much details, we only consider the cases $\delta<\frac{1}{2}$. The analysis for $\delta>\frac{1}{2}$ is similar and is considered in ~\cite{EG}, where a slightly harder problem is solved.

Recall the ideas we presented after Proposition ~\ref{Prop:ThreeTerm}. In what follows we carry out the ideas in details.

\subsection{Proof of ~\eqref{eq:Ga2}}\label{subsec:Ga2}
It is easy to estimate the second factor in the definition of $\Gamma_2:$ Recall that we only consider the case $\delta< \frac{1}{2}.$
\begin{equation}\label{eq:minorPart}
|\int_{0}^{t} (1-\chi_{T}(s_1))\ q_{s_{1}}\ ds_1|\lesssim (1+t)^{\frac{1}{2}-\delta} \|q\|_{\delta,T}.
\end{equation} The crucial step is to prove that
\begin{equation}\label{eq:keyObser}
|Z \int_{0}^{t} K(t-s) Re\langle W, e^{i\frac{\Delta}{2}  s}W\rangle\ ds|\lesssim (1+t)^{-\frac{3}{2}}.
\end{equation} Estimates ~\eqref{eq:minorPart} and ~\eqref{eq:keyObser} obviously imply ~\eqref{eq:Ga2}.

In what follows we prove ~\eqref{eq:keyObser}. Using the fact that $K(t)=0,$ for $t<0,$ in ~\eqref{eq:zero}, applying Fourier transformation, using ~\eqref{eq:FTconv} and applying inverse Fourier transformation we find that
\begin{align}\label{eq:estI2}
I:=&\int_{0}^{t} K(t-s) ZRe\langle W, e^{i\frac{\Delta}{2}  s}W\rangle\ ds\\
=&\int_{0}^{\infty} K(t-s) ZRe\langle W, e^{i\frac{\Delta}{2}  s}W\rangle\ ds\\
=&\frac{1}{2\pi}\int_{-\infty}^{\infty} \frac{F(k)}{ik+ZG(k+i0)} e^{-ikt}\ dk,
\end{align}
where $F(k)$ is defined by
\begin{align*}
F(k)&:=Z\int_{0}^{\infty} e^{iks} Re\langle W, e^{i\frac{\Delta}{2}  s}W\rangle\ ds\\
&=\frac{Z}{2}\int_{0}^{\infty} e^{iks}[ \langle W, e^{i\frac{\Delta}{2}  s}W\rangle+\langle W, e^{-i\frac{\Delta}{2} s}W\rangle]\ ds\\
&=-Z[\langle (-i\Delta +2ik-0)^{-1}W, W\rangle+\langle (i\Delta +2ik-0)^{-1}W, W\rangle]\\
&=-ZG(k+i0),
\end{align*}
and $G(k+i0)$ is defined in ~\eqref{eq:difGk}.

By ~\eqref{eq:Gki0} the function $\frac{Z G(k+i0)}{ik+ZG(k+i0)}$ is smooth on $\mathbb{R}\backslash\{0\}$, and around $k=0$ it has the expression $$\frac{Z G(k+i0)}{ik+ZG(k+i0)}=-1+C\frac{1}{Z} k^{\frac{1}{2}}+O(k)$$ for some constant $C$. Putting this into ~\eqref{eq:estI2}, and integrating by parts in the variable $k$, we obtain
\begin{align}
|I|&=\frac{1}{2\pi} t^{-1} |\int_{-\infty}^{\infty} e^{-ikt} \partial_{k}[\frac{Z G(k+i0)}{ik+ZG(k+i0)}]\ dk|\nonumber\\
&\leq  \frac{1}{2\pi} t^{-1}[A_1+A_2+A_3]\label{eq:decomI}
\end{align}
with $$A_1:=|\int_{-\infty}^{\infty} e^{-ikt} \frac{1}{ik+ZG(k+i0)}\ dk |,$$
$$A_2:=|\int_{-\infty}^{\infty} e^{-ikt} \frac{k}{[ik+ZG(k+i0)]^2}\ dk |,$$
$$A_3:=|\int_{-\infty}^{\infty} e^{-ikt} \frac{Zk \partial_{k}G(k+i0)}{[ik+ZG(k+i0)]^2}\ dk |.$$ It is easy to see that $\frac{1}{ik+ZG(k+i0)}=O(|k|^{-\frac{1}{2}}),\ \frac{k}{[ik+ZG(k+i0)]^2}=O(1)$ and $\frac{Zk\partial_{k}G(k+i0)}{[ik+ZG(k+i0)]^2}=O(|k|^{-\frac{1}{2}}),$ in a neighborhood of $k=0.$ This implies that
\begin{equation}
\sum_{k=1}^{3} A_{k}=O(t^{-\frac{1}{2}}),\ \text{or equivalently},\ I=O(t^{-\frac{3}{2}}),
\end{equation} which, together with the estimates above, yields the desired estimate ~\eqref{eq:keyObser}.
\begin{flushright}
$\square$
\end{flushright}

\subsection{Proof of ~\eqref{eq:Ga1}}\label{subsec:Ga1} Recall the asymptotic forms of $K(t)$ and $Re \langle W,  e^{i\frac{\Delta}{2}t} W\rangle$ in ~\eqref{eq:estKt} and ~\eqref{eq:asymp}, respectively. We define functions $\tilde{K}$, $\tilde{M}$ and $\tilde{\Gamma}_{1}$ to approximate these functions and $\Gamma_{1}$:
\begin{equation}\label{eq:tildeKM}
\begin{array}{lll}
Z\tilde{K}(t)&:=& t^{-\frac{1}{2}}\frac{1}{2} \pi^{-\frac{5}{2}},\\
& &\\
\tilde{M}(t)&:=&-2 t^{-\frac{3}{2}}\pi^{\frac{3}{2}}
\end{array}
\end{equation} and
\begin{equation}\label{eq:estTildeGa1}
\tilde{\Gamma}_{1}(t):= -Z\int_{0}^{t} [\tilde{K}(t-s)-\tilde{K}(t)] \tilde{M}(s) \int_{s}^{t} q_{s_{1}}\ [1- \chi_{T}(s_1)]\  ds_{1}\  ds.
\end{equation}

In the following we prove a sharp estimate on $\tilde\Gamma_{1}$ and prove that $\Gamma_1-\tilde\Gamma_{1}$ is negligibly small because it decays faster.
Recall the definition of $\Omega_1(\delta)$ in ~\eqref{eq:dif12}, and recall that we only consider the case $\delta<\frac{1}{2}.$
\begin{lemma} The function $\tilde{\Gamma}_{1}$ is estimated by
\begin{equation}\label{eq:estTideGam1}
|\tilde{\Gamma}_{1}|(t)\leq t^{-\frac{1}{2}-\delta} \Omega_1(\delta)\|q_{t}\|_{\delta,T},
\end{equation} and there exists a constant $\epsilon(T),$ with $\displaystyle\lim_{T\rightarrow \infty}\epsilon(T)=0$ such that
\begin{equation}\label{eq:estGa1Rem}
|\tilde{\Gamma}_{1}-\Gamma_{1}|(t)\leq  t^{-\frac{1}{2}-\delta}  \epsilon(T) \|q_{t}\|_{\delta,T}.
\end{equation}
\end{lemma}
\begin{proof}
We first prove ~\eqref{eq:estTideGam1}. By direct computation we find
\begin{align*}
|\tilde{\Gamma}_{1}|\leq &\frac{1}{2\pi}\int_{0}^{t} [(t-s)^{-\frac{1}{2}}-t^{-\frac{1}{2}}] s^{-\frac{3}{2}} \int_{s}^{t}|q_{s_1}|\ ds_1 ds\\
\leq & \frac{1}{(1-2\delta)\pi} \int_{0}^{t} [(t-s)^{-\frac{1}{2}}-t^{-\frac{1}{2}}] s^{-\frac{3}{2}} (t^{\frac{1}{2}-\delta}-s^{\frac{1}{2}-\delta})\ ds \ \|q_{t}\|_{\delta,T}\\
=&\frac{1}{(1-2\delta)\pi} \int_{0}^{t} (t-s)^{-\frac{1}{2}}t^{-\frac{1}{2}}\frac{1}{(t-s)^{\frac{1}{2}}+t^{\frac{1}{2}}} s^{-\frac{1}{2}} (t^{\frac{1}{2}-\delta}-s^{\frac{1}{2}-\delta}) ds \ \|q_{t}\|_{\delta,T}.
\end{align*}
Here $(t-s)^{-\frac{1}{2}}-t^{-\frac{1}{2}}$ has been rewritten as $$(t-s)^{-\frac{1}{2}}-t^{-\frac{1}{2}}=(t-s)^{-\frac{1}{2}}t^{-\frac{1}{2}}\frac{s}{(t-s)^{\frac{1}{2}}+t^{\frac{1}{2}}}.$$ Changing variable, $s=tr,$ we see that
$$|\tilde\Gamma_{1}|\leq t^{-\frac{1}{2}-\delta} \Omega_1 (\delta)\ \|q_{t}\|_{\delta,T}.$$

We proceed to estimating $\Gamma_{1}-\tilde{\Gamma}_{1}$ and prove ~\eqref{eq:estGa1Rem}.
We divide $\Gamma_{1}-\tilde{\Gamma}_{1}$ into three parts, according to the integration regions:
\begin{equation}
\Gamma_{1}-\tilde{\Gamma}_{1}=I_{1}+I_{2}+I_{3},
\end{equation} with $I_{k},\ k=1,2,3,$ defined by
\begin{align*}
I_{1}:=&Z\int_{0}^{T^{\frac{1}{3}}} [K(t-s)-K(t)]Re \langle W, e^{i\frac{\Delta}{2}  s}W\rangle \int_{s}^{t}(1-\chi_{T}(s_{1})) q_{s_{1}}\ ds_{1} ds\\
 &-Z\int_{0}^{T^{\frac{1}{3}}} [\tilde{K}(t-s)-\tilde{K}(t)]\tilde{M}(s) \int_{s}^{t}(1-\chi_{T}(s_{1})) q_{s_{1}}\ ds_{1} ds,
\end{align*}
\begin{align*}
I_{2}:=&Z\int_{T^{\frac{1}{3}}}^{t-T^{\frac{1}{3}}} [K(t-s)-K(t)]Re \langle W, e^{i\frac{\Delta}{2}  s}W\rangle \int_{s}^{t}(1-\chi_{T}(s_{1})) q_{s_{1}}\ ds_{1} ds\\
 &-Z\int_{T^{\frac{1}{3}}}^{t-T^{\frac{1}{3}}} [\tilde{K}(t-s)-\tilde{K}(t)]\tilde{M}(s) \int_{s}^{t}(1-\chi_{T}(s_{1})) q_{s_{1}}\ ds_{1} ds\\
=& Z\int_{T^{\frac{1}{3}}}^{t-T^{\frac{1}{3}}} [(K(t-s)-K(t))-(\tilde{K}(t-s)-\tilde{K}(t))]Re \langle W, e^{i\frac{\Delta}{2}  s}W\rangle \int_{s}^{t}(1-\chi_{T}(s_{1})) q_{s_{1}}\ ds_{1} ds\\
&+Z\int_{T^{\frac{1}{3}}}^{t-T^{\frac{1}{3}}} [\tilde{K}(t-s)-\tilde{K}(t)][Re\langle W, e^{i\frac{\Delta}{2}  s}W\rangle-\tilde{M}(s)] \int_{s}^{t}(1-\chi_{T}(s_{1})) q_{s_{1}}\ ds_{1} ds
\end{align*}
 and
\begin{align*}
I_{3}&:=Z\int_{t-T^{\frac{1}{3}}}^{t} [K(t-s)-K(t)]Re \langle W, e^{i\frac{\Delta}{2}  s}W\rangle \int_{s}^{t}(1-\chi_{T}(s_{1})) q_{s_{1}}\ ds_{1} ds\\
& -Z\int_{t-T^{\frac{1}{3}}}^{t} [\tilde{K}(t-s)-\tilde{K}(t)]\tilde{M}(s) \int_{s}^{t}(1-\chi_{T}(s_{1})) q_{s_{1}}\ ds_{1} ds
\end{align*}

We first analyze $I_{1}.$ When $s\leq T^{\frac{1}{3}}$ and $t\geq T$ we use ~\eqref{eq:estKt} to obtain
\begin{align*}
 &|K(t-s)-K(t)|,\  |\tilde{K}(t-s)-\tilde{K}(t)|\\
\lesssim &t^{-\frac{1}{2}} (t-s)^{-\frac{1}{2}}\frac{s}{t^{\frac{1}{2}}+(t-s)^{\frac{1}{2}}}+ [(t-s)^{-1}-t^{-1}]+(t-s)^{-\frac{3}{2}}\\
\lesssim &t^{-\frac{3}{2}} (1+s).
\end{align*}
Consequently
$$
|K(t-s)-K(t)||Re \langle W, e^{i\frac{\Delta}{2}  s}W\rangle |+ |\tilde{K}(t-s)-\tilde{K}(t)||\tilde{M}(t)|
\lesssim t^{-\frac{3}{2}} s^{-\frac{1}{2}}.
$$
Plugging this into the definition of $I_{1},$ we obtain
\begin{equation}\label{eq:estI1}
|I_{1}|\lesssim t^{-1-\delta}\int_{0}^{T^{\frac{1}{3}}} s^{-\frac{1}{2}}\ ds \ \|q\|_{\delta,T}= t^{-1-\delta } T^{\frac{1}{6}} \|q\|_{\delta,T}\leq T^{-\frac{1}{3}} t^{-\frac{1}{2}-\delta}\|q\|_{\delta,T}.
\end{equation}

Now we turn to $I_{2}$. In the region $[T^{\frac{1}{3}},t-T^{\frac{1}{3}}]$ the functions $\tilde{K}(t)$ and $\tilde{M}(t)$ are good approximations of $K(t)$ and $Re\langle W,  e^{i\frac{\Delta}{2}t} W\rangle$. Specifically
\begin{align*}
|K(t-s)-\tilde{K}(t-s)|\lesssim& (1+t-s)^{-1}\\
|K(t)-\tilde{K}(t)|\lesssim &(1+t)^{-1}\\
|Re\langle W, e^{i\frac{\Delta}{2}  s}W\rangle -\tilde{M}(s)|\lesssim & s^{-\frac{5}{2}}.
\end{align*} By direct computation,
\begin{equation}\label{eq:estI22}
I_{2}
\lesssim  t^{-1-\delta}\|q\|_{\delta,T}\leq  T^{-\frac{1}{6}}t^{-\frac{1}{2}-\delta}\|q\|_{\delta,T}.
\end{equation}

Concerning $I_{3}$, it is easy to find that in the region $s\geq t-T^{\frac{1}{3}}$ and $t\geq T$, $\langle W, e^{i\frac{\Delta}{2}  s}W\rangle =O(t^{-\frac{3}{2}}).$ Hence
\begin{align}
|I_{3}|\lesssim &\int_{t-T^{\frac{1}{3}}}^{t} (|t-s|^{-\frac{1}{2}}+t^{-\frac{1}{2}})\ ds\ t^{-1-\delta} \|q\|_{\delta,T}\nonumber\\
\lesssim & T^{\frac{1}{6}} t^{-1-\delta}\|q\|_{\delta,T}\nonumber\\
\leq & T^{-\frac{1}{3}}t^{-\frac{1}{2}-\delta} \|q\|_{\delta,T}.\label{eq:estI3}
\end{align}

We complete the proof by collecting the estimates above and adopting an appropriate definition of $\epsilon(T).$
\end{proof}
\subsection{Proof of ~\eqref{eq:Ga3}}
Similarly as for $\Gamma_1,$ we start with retrieving the `main' contribution to $\Gamma_3$. We define a new function $\tilde{V}$ to
approximate the function $2Re\langle W, (i\Delta)^{-1} e^{i\frac{\Delta}{2}t} W\rangle$
when $t$ is large (see ~\eqref{eq:asymp}): $$\tilde{V}:= 4t^{-\frac{1}{2}}\pi^{\frac{3}{2}}.$$ We then define an approximation, $\tilde{\Gamma}_{3},$ of $\Gamma_3$ by setting
\begin{equation}
\tilde{\Gamma}_{3}:=Z\tilde{K}(t)\int_{0}^{t} [\tilde{V}(t-s)-\tilde{V}(t)] (1-\chi_{T}(s)) q_{s}\ ds,
\end{equation} where, $\tilde{K}(t)$ has been defined in ~\eqref{eq:tildeKM}.

Recall the definitions of $\|q\|_{\delta,T}$ and $\Omega_2(\delta)$ in ~\eqref{eq:norm} and ~\eqref{eq:dif12}, respectively, and recall that we only consider $\delta<\frac{1}{2}$.
\begin{proposition} There exists a time $T>0$ such that, for any $t\geq T$
\begin{equation}\label{eq:estTildeGa3}
|\tilde{\Gamma}_{3}|\leq t^{-\frac{1}{2}-\delta}\Omega_2(\delta)\|q\|_{\delta,T}.
\end{equation}  The remainder $\Gamma_3-\tilde{\Gamma}_{3}$ is estimated by
\begin{equation}\label{eq:differ}
|\Gamma_{3}-\tilde{\Gamma}_{3}|\leq t^{-\frac{1}{2}-\delta}\epsilon(T)\|q\|_{\delta,T}
\end{equation} where $\epsilon(T)$ is a small constant satisfying $\displaystyle\lim_{T\rightarrow}\epsilon(T)=0.$
\end{proposition}
\begin{proof} We start with a proof of ~\eqref{eq:estTildeGa3}.
By direct computation $$
\begin{array}{lll}
|\tilde{\Gamma}_{3}|&\leq &t^{-\frac{1}{2}}\frac{1}{\pi} \int_{0}^{t}[(t-s)^{-\frac{1}{2}}-t^{-\frac{1}{2}}] s^{-\frac{1}{2}-\delta}\ ds \|q\|_{\delta,T}\\
& &\\
&=& t^{-1} \frac{1}{\pi}\int_{0}^{t}\frac{s}{(t-s)^{\frac{1}{2}}+t^{\frac{1}{2}}}(t-s)^{-\frac{1}{2}} s^{-\frac{1}{2}-\delta}\ ds \|q\|_{\delta,T}.
\end{array}
$$ Changing variables $s:=t r$ we obtain the desired estimate
$$|\tilde{\Gamma}_{3}|\leq t^{-\frac{1}{2}-\delta}\frac{1}{\pi}\int_{0}^{1}(1-r)^{-\frac{1}{2}}\frac{1}{1+(1-r)^{\frac{1}{2}}}r^{\frac{1}{2}-\delta}\ dr \|q\|_{\delta,T}= t^{-\frac{1}{2}-\delta} \Omega_2(\delta)\|q\|_{\delta,T}.$$

To prove ~\eqref{eq:differ} we transform $\Gamma_{3}-\tilde{\Gamma}_{3}$ into a more convenient form:
\begin{equation}\label{eq:differen3}
\begin{array}{lll}
\Gamma_{3}-\tilde{\Gamma}_{3}&=&Z[K(t)-\tilde{K}(t)]\int_{0}^{t}[V(t-s)-V(t)](1-\chi_{T}(s)) q_{s}\ ds\\
& &\\
& &+Z\tilde{K}(t)\int_{0}^{t}[V(t-s)-\tilde{V}(t-s)](1-\chi_{T}(s)) q_{s}\ ds\\
& &\\
& &+Z\tilde{K}(t)[\tilde{V}(t)-V(t)]\int_{0}^{t}(1-\chi_{T}(s)) q_{s}\ ds
\end{array}
\end{equation} where the function $V(t)$ is defined by $$V(t):=2Re\langle W, (i\Delta)^{-1} e^{i\frac{\Delta}{2}t} W\rangle.$$
By arguments almost identical to the proof of ~\eqref{eq:estGa1Rem} we conclude that ~\eqref{eq:differ} holds.
\end{proof}
\section{Proof of ~\eqref{eq:estFp}}\label{sec:highorderterm}
In this section we prove inequality ~\eqref{eq:estFp} in Theorem ~\ref{THM:Contraction}. As in Appendix ~\ref{SEC:contraction},
we only consider the case $\delta<\frac{1}{2}.$ The case $\delta>\frac{1}{2}$ is considered in ~\cite{EG}.
\begin{proposition}
Let $Q_1$ and $Q_2$ be defined in Theorem ~\ref{THM:Contraction}, and recall the definitions of $F(P)$ and $B_{k},\ k=0,1,2,$ in ~\eqref{eq:difVecF}.

If $|P_0|\leq T^{-3}$, then
\begin{equation}\label{eq:contracB0}
|B_{0}(Q_1)-B_{0}(Q_2)|(t)\ \ll (1+t)^{-1-\delta}\|Q_1 -Q_2\|_{\delta,T} ,
\end{equation} and, for $k=1,2,$
\begin{equation}\label{eq:ContracB1B2}
|B_{k}(Q_1)-B_{k}(Q_2)|(t)\lesssim (1+t)^{-1-\delta}\|Q_1 -Q_2\|_{\delta,T} (\|Q_1 \|_{\delta,T}+\|Q_2\|_{\delta,T}).
\end{equation}
\end{proposition}
This proposition is proven below.

Now we use this result to prove the desired estimate ~\eqref{eq:estFp}.\\
{\bf{Proof of ~\eqref{eq:estFp}}}

We start by estimating $K(t-s)-K(t)$. Using ~\eqref{eq:estKt}, we obtain
\begin{align}\label{eq:ktks}
|K(t-s)-K(t)|\lesssim &|(1+t-s)^{-\frac{1}{2}}-(1+t)^{-\frac{1}{2}}|+(1+t-s)^{-1}-(1+t)^{-1}+(1+t-s)^{-\frac{3}{2}}\nonumber\\
= & (1+t-s)^{-\frac{1}{2}} (1+t)^{-\frac{1}{2}}\frac{s}{(1+t)^{\frac{1}{2}}+(1+t-s)^{\frac{1}{2}}}\nonumber\\
 &+s(1+t-s)^{-1}(1+t)^{-1}+(1+t-s)^{-\frac{3}{2}}\nonumber\\
\leq &(1+t-s)^{-\frac{1}{2}} (1+t)^{-1} s+(1+t-s)^{-\frac{3}{2}}.
\end{align}
Plugging this estimate, as well as ~\eqref{eq:contracB0} and ~\eqref{eq:ContracB1B2} into $\int_{0}^{t}|K(t-s)-K(t)||B_{k}(Q_1)-B_{k}(Q_2)|(s)\ ds,\ k=0,1,2,$ we obtain the desired estimate ~\eqref{eq:estFp}.
\begin{flushright}
$\square$
\end{flushright}
\subsection{Proof of ~\eqref{eq:contracB0}}
By direct computation and our estimate on $\| e^{i\frac{\Delta}{2}t} \|_{L^{1}\rightarrow L^{\infty}}$ we obtain that
\begin{align*}
|B_{0}(Q_1)-B_{0}(Q_2)|(t)
\lesssim t^{-\frac{3}{2}}\|\beta_0\|_{L^{1}}\| \partial_{x}W^{\int_{T}^{t} [Q_1(s)-Q_2(s)] ds}-\partial_{x}W\|_{L^{1}}
\end{align*}
The right hand side vanishes when $t\leq T$. This, together with $T\gg 1,$ implies that
\begin{align*}
|B_{0}(Q_1)-B_{0}(Q_2)|(t)&=(1+t)^{-\frac{3}{2}}\|\beta_0\|_{L^{1}}\| \int_{T}^{t}\partial_s \partial_{x}W^{\int_{T}^{s} [Q_1(s_1)-Q_2(s_1)] ds_1} \ ds\|_{L^{1}}\\
&\leq (1+t)^{-\frac{3}{2}}\|\beta_0\|_{L^{1}}\|\partial_{x}W\|_{L^1} \int_{T}^{t}|Q_1(s)-Q_2(s)|\ ds \\
&\ll  (1+t)^{-1-\delta} \|Q_1-Q_2\|_{\delta,T}.
\end{align*} Here the hypothesis of smallness of $\|\beta_0\|_{L^{1}}$ in Theorem ~\ref{THM:main} is used, and we recall that only the case $\delta<\frac{1}{2}$ is considered.
\begin{flushright}
$\square$
\end{flushright}

\subsection{Proof of ~\eqref{eq:ContracB1B2}}
In what follows we only estimate the contributions to $B_2$; the estimate of $B_1$ is similar and easier, hence omitted.

We start with transforming $B_{2}$. Using the formula $$W^{X_{s}-X_{t}}-W=-\int_{s}^{t} \partial_{s_{1}} W^{X_{s_{1}}-X_{t}} \ ds_{1}=\int_{s}^{t} [P_{s_{1}}\cdot \partial_{x}]\ W^{X_{s_{1}}-X_{t}} \ ds_{1}$$ we obtain
\begin{equation}
\begin{array}{lll}
B_{2}(P)&=&-2\nu  Re\langle \nabla_{x} W, \ \int_{0}^{t} e^{i\frac{\Delta}{2} (t-s)}(-\Delta)^{-1} [P_{s}\cdot\partial_{x}] \int_{s}^{t}[P_{s_{1}}\cdot\partial_{x}] W^{X_{s_{1}}-X_{t}}\rangle\ ds_{1}ds\\
& &\\
&=&-2\nu  Re\langle \nabla_{x} W, \ \int_{0}^{t} e^{i\frac{\Delta}{2} (t-s)}(-\Delta)^{-1} [P_{s}\cdot\partial_{x}] \int_{s}^{t}[P_{s_{1}}\cdot\partial_{x}] [W^{X_{s_{1}}-X_{t}}-W]\rangle\ ds_{1}ds\\
& &\\
&=&-2\nu Re\langle \nabla_{x} W, \ \int_{0}^{t} e^{i\frac{\Delta}{2} (t-s)}(-\Delta)^{-1} [P_{s}\cdot\partial_{x}] \int_{s}^{t}[P_{s_{1}}\cdot\partial_{x}] \int_{s_{1}}^{t}[P_{s_{2}}\cdot\partial_{x}]W^{X_{s_{2}}-X_{t}}\rangle\ ds_{2}\ ds_{1}\ ds.
\end{array}
\end{equation} Here the fact that the function $W$ is spherically symmetric has been used to see that various terms vanish.

Next, we prove ~\eqref{eq:ContracB1B2} for $k=2$. Among many terms in its expression we only consider $$
\begin{array}{lll}
\tilde{B}&:=&-2\nu Re\langle \nabla_{x} W, \ \int_{0}^{t} e^{i\frac{\Delta}{2} (t-s)}(-\Delta)^{-1} [Q_{1}(s)\cdot\partial_{x}] \int_{s}^{t}[Q_{1}(s_{1})\cdot\partial_{x}]\times\\
& &\\
& & \int_{s_{1}}^{t}[Q_{1}(s_{2})\cdot\partial_{x}][W^{X_{s_{2}}-X_{t}}-W^{\tilde{X}_{s_{2}}-\tilde{X}_{t}}]\rangle\ ds_{2}\ ds_{1}\ ds,
\end{array}
$$
where $X_{t}:= X_0+\int_{0}^{t} Q_{1}{s}$ and $\tilde{X}_{t}:=X_0+\int_{0}^{t} Q_{2}(s)\ ds.$ We rewrite $W^{X_{s_{2}}-X_{t}}-W^{\tilde{X}_{s_{2}}-\tilde{X}_{t}}$ as follows
\begin{align*}
W^{X_{s_{2}}-X_{t}}-W^{\tilde{X}_{s_{2}}-\tilde{X}_{t}}=&-\int_{s_2}^{t} \partial_{z} W^{X_{z}-X_{t}-\tilde{X}_{z}+\tilde{X}_{s_2}}\ dz\\
=&\int_{s_2}^{t}(Q_{1}(z)-Q_2(z))\cdot \nabla_{x} W^{X_{z}-X_{t}-\tilde{X}_{z}+\tilde{X}_{s_2}}\ dz.
\end{align*}
Consequently
\begin{align*}
B_0=&-2\nu Re\langle \nabla_{x} W, \ \int_{0}^{t} e^{i\frac{\Delta}{2} (t-s)}(-\Delta)^{-1} [Q_{1}(s)\cdot\partial_{x}] \int_{s}^{t}[Q_{1}(s_{1})\cdot\partial_{x}]\times\\
& \int_{s_{1}}^{t}[Q_{1}(s_{2})\cdot\partial_{x}]\int_{s_2}^{t}(Q_{1}(z)-Q_2(z))\cdot \nabla_{x} W^{X_{z}-X_{t}-\tilde{X}_{z}+\tilde{X}_{s_2}}\ \rangle\ dz\ ds_{2}\ ds_{1}\ ds
\end{align*}

To obtain a decay estimate we have to consider a function $I_4$ define by
\begin{align*}
I_{4}:=&\langle \nabla_{x} W,  e^{i\frac{\Delta}{2}(t-s)}  (-\Delta)^{-1} (\displaystyle\prod_{k_l\in \{1,2,3\},\ l=1,2,3 }\partial_{x_{k_{l}}}) W^{Y(t,s_2,z)}\rangle \\
=&\langle \nabla_{x} W, e^{i\frac{\Delta}{2}  (t-s)+Y(t,s_2,z) \cdot \partial_{x}} (-\Delta)^{-1} (\displaystyle\prod_{k_l\in \{1,2,3\},\ l=1,2,3 }\partial_{x_{k_{l}}}) W\rangle,
\end{align*}
where $Y(t,s_2,z)\in \mathbb{R}^3$ is defined by $Y(t,s_2,z):=X_{z}-X_{t}-\tilde{X}_{z}+\tilde{X}_{s_2},$ with $s_2\in [s,t]$ and $z\in [s_2,t].$ It is easy to transform $I_{4}$ to a form where Proposition ~\ref{prop:propa}, below, applies, which yields
\begin{equation}\label{eq:I4}
|I_{4}|\lesssim (1+t-s)^{-3}, \ \text{for any} \ t\ \text{and}\ s\leq t.
\end{equation}
Now it is easy to see that, for some constant $C$,
\begin{align}
|\tilde{B}|\lesssim & \int_{0}^{t}(1+t-s)^{-3} s^{-\frac{1}{2}-\delta}\int_{s}^{t} s_{1}^{-\frac{1}{2}-\delta} \int_{s_{1}}^{t} s_{2}^{-\frac{1}{2}-\delta} \int^{t}_{s_2} s_3^{-\frac{1}{2}-\delta}\ ds_3ds_2 ds_1 ds \|Q_1\|_{\delta,T}^3 \|Q_1-Q_2\|_{\delta,T}\\
&= C[t^{\frac{3}{2}-3\delta}\int_{0}^{t} (1+t-s)^{-3} s^{-\frac{1}{2}-\delta}\ ds-3 t^{1-2\delta} \int_{0}^{t} (1+t-s)^{-3} s^{-2\delta}\ ds\\
& + 3t^{\frac{1}{2}-\delta}\int_{0}^{t} (1+t-s)^{-3} s^{\frac{1}{2}-3\delta}\ ds-\int_{0}^{t} (1+t-s)^{-3} s^{1-4\delta}\ ds]\times\\
& \|Q_1\|_{\delta,T}^3 \|Q_1-Q_2\|_{\delta,T}.
\end{align}

To estimate the integrals above we divide the interval $[0,t]$ into two subintervals: $ [0,\frac{1}{2}t]$ and $ [\frac{1}{2}t, t]$, and denote the corresponding contributions by $D_1$ and $D_2$ respectively.
\begin{itemize}
\item[(1)]
Concerning $D_1,$ it is easy to see that $$D_{1}\lesssim t^{-1-4\delta} \|Q_1\|_{\delta,T}^3 \|Q_1-Q_2\|_{\delta,T} $$ using the fact that $(1+t-s)^{-3}=O((1+t)^{-3})$ and direct computation.
\item[(2)]
In the second interval we Taylor-expand the functions $s^{\alpha}=t^{\alpha}[1-\alpha \frac{t-s}{t}+O([\frac{t-s}{t}]^2)],$ with $\alpha= -\frac{1}{2}-\delta,\ -2\delta,\ \frac{1}{2}-3\delta$, and find that the leading terms cancel each other. This implies that
$$D_{2}\lesssim  t^{-1-4\delta}\int_{\frac{1}{2}t}^{t} (1+t-s)^{-1}\ ds\|Q_1\|_{\delta,T}^3 \|Q_1-Q_2\|_{\delta,T} \lesssim  t^{-1-4\delta}lnt \|Q_1\|_{\delta,T}^3 \|Q_1-Q_2\|_{\delta,T}.$$
\end{itemize}
Collecting the above estimates we conclude that
\begin{equation}
\tilde{B}(t)\lesssim t^{-1-2\delta} \|Q_1\|_{\delta,T}^3 \|Q_1-Q_2\|_{\delta,T}.
\end{equation}
To remove the non-integrable singularity at $t=0$, we use the fact that $\tilde{B}(t)=0,$ for $t<T,$ with $T\gg 1.$
\begin{flushright}
 $\square$
\end{flushright}
\subsection{Some Decay Estimates}
The following result has been used in ~\eqref{eq:I4}.
\begin{proposition}\label{prop:propa}
Suppose that the functions $f:\mathbb{R}^{+}\rightarrow \mathbb{C}$, $h:[0,2\pi]\rightarrow \mathbb{C}$ and $y:\mathbb{R}^{+}\rightarrow \mathbb{R}$ are smooth and satisfy the conditions:
$$|y(t)| t^{-\frac{1}{2}}\leq \epsilon_{0},\ \text{for some small constant}\ \epsilon_0>0;$$
$$|\partial_{\rho}^{k}f(\rho)| \leq a_{k} e^{-c_{0}|\rho|}, \ \text{for any }\ k\in \mathbb{N}, \ \text{where}\ a_{k},\ c_{0}>0\ \text{are constants}.$$ Then there exist constants $C_{l},\ l\in \mathbb{N},$ such that
\begin{equation}\label{eq:propa}
|\int_{0}^{2\pi} h(\theta)\int_{0}^{\infty} e^{i\rho^2 t} e^{i\rho cos\theta\ y(t)} \rho^{l} f(\rho^2) \ d\rho d\theta| \leq C_{l} (1+t)^{-\frac{l+1}{2}}.
\end{equation}
\end{proposition}
\begin{proof}
If $y\equiv 0$ the proof is standard, (integration by parts). In the present situation we use integration by parts and the fact that $y(t)$ is appropriately small.

For notational purposes, we define functions $I_{l}$ by
$$I_{l}(t):=\int_{0}^{2\pi} h(\theta)\int_{0}^{\infty} e^{i\rho^2 t} e^{i\rho cos\theta\ y(t)} \rho^{l} f(\rho^2) \ d\rho\  d\theta.$$

We prove by induction that $$|I_{l}|\leq C_{l}(1+t)^{-\frac{l+1}{2}}.$$

{\bf{Step 1}}
For $l=0$, we change variables, $\rho^2 t=: r,$ to obtain
\begin{equation}\label{eq:split}
\begin{array}{lll}
I_{0}&=& \frac{1}{2} t^{-\frac{1}{2}}\int_{0}^{2\pi} h(\theta)\int_{0}^{\infty} r^{-\frac{1}{2}} e^{i r} e^{i r^{\frac{1}{2}} cos\theta y(t) t^{-\frac{1}{2}}} f(t^{-1} r) \ dr d\theta \\
& &\\
&=&\frac{1}{2} t^{-\frac{1}{2}}[J_1+J_{2} ]
\end{array}
\end{equation} with $$J_{1}:=\int_{0}^{2\pi} h(\theta)\int_{0}^{1} r^{-\frac{1}{2}} e^{i r} e^{i r^{\frac{1}{2}} cos\theta y(t) t^{-\frac{1}{2}}} f(t^{-1} r) \ dr d\theta$$
and $$J_{2}:=\int_{0}^{2\pi} h(\theta)\int_{1}^{\infty} r^{-\frac{1}{2}} e^{i r} e^{i r^{\frac{1}{2}} cos\theta y(t) t^{-\frac{1}{2}}} f(t^{-1} r) \ dr d\theta.$$
$J_{1}$ is clearly bounded.

To bound $J_{2}$ an obvious obstacle is that the function $r^{-\frac{1}{2}}\not\in L^{1}[0,\infty)$. The way out is to integrate by parts, using that $$-i \frac{1}{1+ \frac{1}{2} r^{-\frac{1}{2}} cos\theta y(t) t^{-\frac{1}{2}}}\partial_{r}[e^{i r} e^{i r^{\frac{1}{2}} cos\theta y(t) t^{-\frac{1}{2}}}]= e^{i r} e^{i r^{\frac{1}{2}} cos\theta y(t) t^{-\frac{1}{2}}}$$ and to use that the function $$g(r,\theta):=-i\frac{1}{1+ \frac{1}{2} r^{-\frac{1}{2}} cos\theta y(t) t^{-\frac{1}{2}}}$$ is uniformly bounded, for $r\geq 1$, as follows from the hypotheses of the proposition. Then
$$|J_{2}|\leq \int_{0}^{2\pi} |h(\theta)| |g(1,\theta)| |f(t^{-1})| \ d\theta+\int_{0}^{2\pi} |h(\theta)|\int_{1}^{\infty}\ |\partial_{r} [r^{-\frac{1}{2}} g(r,\theta) f(t^{-1} r)]|\ dr\ d\theta.$$
The first term on the right hand side is obviously bounded. For the second term, it is not difficult to derive that
$$|\partial_{r} [r^{-\frac{1}{2}} g(r,\theta) f(t^{-1} r)]|\lesssim r^{-\frac{3}{2}}+  t^{-1} r^{-\frac{1}{2}}| f^{'}(t^{-1} r)|.$$ Using our assumptions on the function $f$ we find that $$t^{-1} \int_{1}^{\infty} r^{-\frac{1}{2}}| f^{'}(t^{-1} r)|\ dr\lesssim t^{-\frac{1}{2}}.$$

Collecting the estimates above, $|J_{2}|$ is seen to be bounded. This, together with ~\eqref{eq:split} and the fact $J_{1}$ is bounded, implies the desired result.

{\bf{Step 2}}\\
In the second step of the induction we assume that ~\eqref{eq:propa} holds for any $k\leq l$.

Now we estimate $I_{l+1}.$ The idea is to represent $I_{l+1}$ into a linear combination of $I_{l-1}$ and $I_{l}$ by performing integration by parts on certain variables; (with $I_{-1}=0$).

Integrating by parts, and using the identity $-\frac{i}{2} t^{-1} \partial_{\rho} e^{i\rho^2 t}= \rho e^{i\rho^2 t}$, we find that
\begin{align*}
I_{l+1}=& it^{-1}\int_{0}^{\pi} h(\theta)\int_{0}^{2\infty} e^{i\rho^2 t} \partial_{\rho}[e^{i\rho cos\theta\ y(t)} \rho^{l} f(\rho^2)] \ d\rho d\theta \\
=& t^{-1} y(t) K_{1}+ t^{-1}K_{2}
\end{align*}
with $$K_{1}:=-\frac{1}{2}\int_{0}^{2\pi} cos\theta h(\theta)\int_{0}^{\infty} e^{i\rho^2 t} e^{i\rho cos\theta\ y(t)} \rho^{l} f(\rho^2) \ d\rho d\theta$$ and $$K_{2}:=\frac{i}{2}\int_{0}^{2\pi} h(\theta)\int_{0}^{\infty} e^{i\rho^2 t} e^{i\rho cos\theta\ y(t)}\partial_{l} [\rho^{l} f(\rho^2)] \ d\rho d\theta.$$

Notice $K_{1}$ and $K_{2}$ are of the form of $I_{l}$ or $I_{l-1}$ (if $l\geq 1$), after defining appropriate functions $f_{new}(\rho^2)$ and $h_{new}(\theta)$. Hence the induction hypotheses on $I_{k}, k\leq l$ imply that $$|I_{l+1}|\lesssim |t^{-1}y(t)| t^{-\frac{l+1}{2}}+ t^{-\frac{l+2}{2}}\lesssim t^{-\frac{l+2}{2}},$$ which is the desired estimate on $I_{l+1}.$

Thus, ~\eqref{eq:propa} holds for any $l\in \mathbb{Z}^{+}.$
\end{proof}
\section{Proof of Theorem ~\ref{THM:smallness}}\label{Sec:nonlin}
We reformulate this theorem in the form of the following lemma. Recall the constant $\epsilon_{0}(T)$ in Theorem ~\ref{THM:wellposed} and the definition of $F_{\chi_{T}P}$ in ~\eqref{eq:difVecF}.
\begin{lemma}
If $T$ is sufficiently large and $|P_{0}|, \ \|\langle x\rangle^{3}\beta_{0}\|_{2}\leq \epsilon_{0}(T)$ then, for any time $t\geq T,$
\begin{equation}\label{eq:gam0}
|A(\chi_{T}P)|\leq \epsilon(T) (1+t)^{-\frac{1}{2}-\delta};
\end{equation}
\begin{equation}\label{eq:estB0}
|\int_{0}^{t}(K(t-s)-K(t))\ B_{0}(\chi_{T}P)(s)\ ds|\lesssim \epsilon(T)(1+t)^{-\frac{3}{2}} \|\langle x\rangle^{4}\beta_0\|_{2};
\end{equation} and, for $k=1,2,$
\begin{equation}\label{eq:estB22}
|\int_{0}^{t}(K(t-s)-K(t))\ B_{k}(\chi_{T}P)(s)\ ds|\leq \epsilon(T) t^{-1-2\delta},
\end{equation}
where $\epsilon(T)$ is a small constant satisfying $\displaystyle\lim_{T\rightarrow \infty}\epsilon(T)=0.$
\end{lemma}
\begin{proof}
We start with proving ~\eqref{eq:gam0}.
By its definition in ~\eqref{eq:secQt}, the function $A(\chi_{T}P)$ takes the form
\begin{align*}
A(\chi_{T}P)=&-Z\int_{0}^{t} [K(t-s)-K(t)] Re\langle W, e^{i\frac{\Delta}{2}  s}W\rangle \int_{s}^{t} P_{s_{1}}\  \chi_{T}(s_1)\  ds_{1}\  ds\\
 &+Z\int_{0}^{t} K(t-s) Re\langle W, e^{i\frac{\Delta}{2}  s}W\rangle ds\ \int_{0}^{t} P_{s_{1}} \chi_{T}(s_1)\ ds_{1}\\
 &+2Z K(t) Re\langle W, (i\Delta)^{-1} \int_{0}^{t}[ e^{i\frac{\Delta}{2}(t-s)} - e^{i\frac{\Delta}{2}t} ] \ P_{s}\ \chi_{T}(s)\ ds\ W\rangle.
\end{align*}

Using ~\eqref{eq:FiniInte} we obtain
\begin{equation}\label{eq:estGam0}
\begin{array}{lll}
|A(\chi_{T}P)|&\leq &Z T^{-1}[\int_{0}^{T} |K(t-s)-K(t)| |Re\langle W, e^{i\frac{\Delta}{2}  s}W\rangle| \  ds\\
& &\\
& &+|\int_{0}^{t} K(t-s) Re\langle W, e^{i\frac{\Delta}{2}  s}W\rangle ds|\\
& &\\
& &+ |K(t)| \int_{0}^{T}|Re\langle W, (i\Delta)^{-1}[ e^{i\frac{\Delta}{2}(t-s)} - e^{i\frac{\Delta}{2}t} ]  W\rangle |\ ds].
\end{array}
\end{equation}

As proven in ~\eqref{eq:keyObser}, the second term on the right hand side is of order $t^{-\frac{3}{2}}.$

We now turn to the remaining two terms. For $Re\langle W, (i\Delta)^{-1}[e^{i\frac{\Delta}{2} (t-s)}- e^{i\frac{\Delta}{2}t} ]W\rangle,$ in the last line, we use estimates similar to ~\eqref{eq:ktks} to conclude that
$$
|Re\langle W, (i\Delta)^{-1}[e^{i\frac{\Delta}{2} (t-s)}- e^{i\frac{\Delta}{2}t} ]W\rangle|\lesssim (1+t-s)^{-\frac{1}{2}} (1+t)^{-1} s+(1+t-s)^{-1}.
$$
Putting this and ~\eqref{eq:ktks} back into ~\eqref{eq:estGam0}, we obtain
\begin{align*}
|A(\chi_{T}P)|\leq & T^{-1} [(1+t)^{-1}\int_{0}^{T} (1+t-s)^{-\frac{1}{2}} \frac{s}{(1+s)^{\frac{3}{2}}}\ ds+\int_{0}^{t} (1+t-s)^{-1} (1+s)^{-\frac{3}{2}}\ ds\\
 &+(1+t)^{-\frac{3}{2}} \int_{0}^{T} (1+t-s)^{-\frac{1}{2}} s \ ds+(1+t)^{-\frac{1}{2}} \int_{0}^{T}(1+t-s)^{-1} \ ds]\\
\leq &T^{-\frac{1}{3}}(1+t)^{-\frac{1}{2}-\delta},
\end{align*}
where, to obtain the last line, we consider two regimes, $t\in [T,2T]$ and $t\geq 2T.$

The proof of ~\eqref{eq:gam0} is completed by setting $\epsilon(T)=T^{-\frac{1}{3}}.$


To prove ~\eqref{eq:estB0} we use standard arguments to conclude that $$|B_{0}(\chi_{T}P)|\lesssim \|\langle x\rangle^{-4}e^{-i\Delta t}\nabla_{x} W^{\chi_{T}P}\|_{2}\|\langle x\rangle^{4}\beta_0\|_{2}\lesssim (1+t)^{-2}\|\langle x\rangle^{4}\beta_{0}\|_{2}.$$ This, together with ~\eqref{eq:ktks} and the fact $t\geq T$, implies ~\eqref{eq:estB0}.

To prove ~\eqref{eq:estB22} we only estimate $B_1$, the estimate on $B_2$ being very similar.

Recall the definition of $P_{t}$.
Applying the formula $$W^{X_{0}-X_{s}}-W= \int_{0}^{s} \partial_{s_{1}} W^{X_{0}-X_{s_1}}\ ds_1=
\int_{0}^{s}[P_{s_1}\cdot \partial_x ]W^{X_{0}-X_{s_1}}\ ds_1 $$ we rewrite the expression for $B_1:$
$$B_{1}(P)=2\nu  Re \langle \nabla_{x} W,\  e^{i\frac{\Delta}{2}t}  (-\Delta)^{-1} \int_{0}^{t}[P_s\cdot \partial_x] \int_{0}^{s} [P_{s_{1}}\cdot\partial_x] W^{X_0-X_{s_{1}}}\rangle \ ds_1 ds.$$
Define $$\tilde{X}_{t}:=X_0+\int_{0}^{t}\chi_{T}P_{s}\ ds.$$ Then
\begin{align*}
&B_{1}(\chi_{T} P)\\
=&2\nu  Re \langle \nabla_{x} W,\  e^{i\frac{\Delta}{2}t}  (-\Delta)^{-1} \int_{0}^{t}[P_s\cdot \partial_x] \int_{0}^{s} [P_{s_{1}}\cdot\partial_x] W^{X_0-\tilde{X}_{s_{1}}}\rangle \ ds_1 ds.
\end{align*}

By the estimate on $\chi_{T}P_{t}$ in ~\eqref{eq:FiniInte}, it is easy to see that
\begin{equation}\label{eq:finiteInter}
\int_{0}^{t} \chi_{T}|P_{s}|\ ds=\int_{0}^{T} \chi_{T}|P_{s}|\ ds\leq T^{-1}.
\end{equation}

By a similar argument as in ~\eqref{eq:I4}, we may apply Proposition ~\ref{prop:propa} to find that $$|\langle \nabla_{x} W,\  e^{i\frac{\Delta}{2}t}  (-\Delta)^{-1}  \partial_{x_{k}}\partial_{x_{l}} W^{X_0-\tilde{X}_{s_{1}}}\rangle |\lesssim (1+t)^{-2},$$ for any $s_{1}\in [0,t]$.
This, together with ~\eqref{eq:FiniInte}, implies that
$$
|B_{1}(\chi_{T}P)|\lesssim  (1+t)^{-2}T^{-2}
$$ which, together with ~\eqref{eq:ktks}, implies ~\eqref{eq:estB22}.

\end{proof}
\section{Some Heuristic Ideas Underlying ~\eqref{eq:finalForm}}\label{sub:tacitIdeas}\label{sec:Kideas}
We emphasize that the non-rigorous discussions in this subsection is NOT used in any other parts of the paper.

In this section we present the ideas underlying Eq.~\eqref{eq:finalForm}, which is the key equation in our proof of the most important result, Theorem ~\ref{THM:Contraction}.

To save space we only consider the linear part of ~\eqref{eq:linear}, which corresponds to the following equation
\begin{align}\label{eq:integro}
\dot{q}_{t}=&Z Re\langle W, \  e^{i\frac{\Delta}{2}t} W\rangle \int_{0}^{t} q_{s}\ ds-Z Re \langle W,\ \int_{0}^{t}  e^{i\frac{\Delta}{2}(t-s)} \ q_{s}\ ds\ W\rangle\\
q_0=&1.\nonumber
\end{align}
By repeating the arguments used to prove Theorems ~\ref{THM:Contraction} and ~\ref{THM:smallness} we prove that
the solution $q_t$ belongs to the Banach space $B_{\delta,T}.$

We divide our discussion into two parts: In subsection ~\ref{subsec:ideas}, we present the ideas behind constructing ~\eqref{eq:finalForm}, assuming that $|q_t|\leq const\ t^{-\frac{1}{2}-\delta}$ for some $\delta>0$. In subsection ~\ref{SEC:bestconstants}, we present a heuristic argument to show that $|q_t|\leq const\ t^{-\frac{1}{2}-\delta},$ for some $\delta>0$.
\subsection{Ideas underlying ~\eqref{eq:finalForm}}\label{subsec:ideas}
Recall that ~\eqref{eq:finalForm} is obtained by subtracting from ~\eqref{eq:Durhamel} the product of ~\eqref{eq:integrate} with
the function $K(t)$.

We define a constant $$C_0:=q_0-2Z\int_{0}^{\infty} Re\langle W,\ (i\Delta)^{-1} e^{i\frac{\Delta}{2}  s}W\rangle\ q_{s}\ ds.$$
\begin{lemma}\label{LM:discussion1}
 If the function $q$ satisfies the estimate $|q_{t}|\leq ct^{-\frac{1}{2}-\delta},$ for some $\delta>0$, then
\begin{equation}\label{eq:condition}
C_0=0.
\end{equation}
\end{lemma}
This lemma will be proven shortly.

For technical reasons, it is hard to prove that $|q_t|\leq C t^{-\frac{1}{2}-\delta},$ even after imposing $C_0=0$ on ~\eqref{eq:Durhamel}. Instead we search for a new equation containing $C_0$. A natural candidate is ~\eqref{eq:integrate}, which is obtained by integrating both sides of ~\eqref{eq:linear} from $0$ to $t$. The key observation is:
\begin{lemma}\label{LM:secon}
If the function $q$ satisfies the estimate $|q_{t}|\leq ct^{-\frac{1}{2}-\delta},$ for some $\delta>0$, then ~\eqref{eq:integrate} can be rewritten as
\begin{equation}
 q_t=C_0+O(t^{-\delta}).
\end{equation}
\end{lemma}
The proof of this lemma is almost identical to that of Lemma ~\ref{LM:discussion1}.

By Lemmas ~\ref{LM:discussion1} and ~\ref{LM:secon} it is natural to consider Eq.~\eqref{eq:finalForm}, as we did.

\begin{flushleft}
{\bf{Proof of Lemma ~\ref{LM:discussion1}}}
\end{flushleft}
The key point in proving the lemma is to show that
\begin{equation}\label{eq:asymptotics}
\int_{0}^{t}K(t-s) Re\langle W,\ e^{i\frac{\Delta}{2}  s}W\rangle \int_{0}^{s}q_{s_1} ds_1 ds
=K(t)\int_{0}^{t} Re\langle W,\ e^{i\frac{\Delta}{2}  s}W\rangle \int_{0}^{s}q_{s_1} ds_1 ds+O(t^{-\frac{1}{2}-\delta}).
\end{equation} Suppose this holds. Then by integrating by parts in the variable $s$ and
the assumption that $q_t\leq Ct^{-\frac{1}{2}-\delta},$ we find
\begin{align*}
& \int_{0}^{t}K(t-s) Re\langle W,\ e^{i\frac{\Delta}{2}  s}W\rangle \int_{0}^{s}q_{s_1} ds_1 ds\\
=&-2K(t)\int_{0}^{t} Re\langle W,\ (i\Delta)^{-1} e^{i\frac{\Delta}{2}  s}W\rangle\ q_{s}\ ds+O(t^{-\frac{1}{2}-\delta})\\
=&-2K(t)\int_{0}^{\infty} Re\langle W,\ (i\Delta)^{-1} e^{i\frac{\Delta}{2}  s}W\rangle\ q_{s}\ ds+O(t^{-\frac{1}{2}-\delta})
\end{align*} which, together with the estimate $Re\langle W,\ (i\Delta)^{-1}  e^{i\frac{\Delta}{2}t} W\rangle=O(t^{-\frac{1}{2}})$ proved in ~\eqref{eq:asymp}, obviously implies ~\eqref{eq:condition}.

Now we turn to ~\eqref{eq:asymptotics}.
Define a function $f:[0,\infty)\rightarrow \mathbb{R}$ by $$f(t):=Re\langle W,\  e^{i\frac{\Delta}{2}t} W\rangle \int_{0}^{t} q_{s_1} ds_1.$$ Our assumption on $q_t$ and the estimate $Re\langle W,\  e^{i\frac{\Delta}{2}t} W\rangle=O(t^{-\frac{3}{2}})$ in ~\eqref{eq:asymp} imply that there exists a constant $C_1$ such that
$$|f(t)|\leq C_1 (1+t)^{-1-\delta}.$$ To prove ~\eqref{eq:asymptotics} it is sufficient to prove that
\begin{equation}\label{eq:secondStep}
\int_{0}^{t}|K(t-s)-K(t)| |f(s)|\ ds=O(t^{-\frac{1}{2}-\delta}).
\end{equation}
To see this we use the asymptotics of the function $K$ in ~\eqref{eq:estKt} and obtain that
\begin{align*}
& \int_{0}^{t}|K(t-s)-K(t)| |f(s)|\ ds\\
\lesssim & \int_{0}^{t} [(t-s)^{-\frac{1}{2}}-t^{-\frac{1}{2}}] (1+s)^{-1-\delta}\ ds+\int_{0}^{t}[(1+t-s)^{-\frac{3}{2}}+(1+t)^{-\frac{3}{2}}] (1+s)^{-1-\delta}\ ds.
\end{align*}
The second term on the right hand side is of order $t^{-1-\delta}.$ For the first term, we use the observation that
$$0<(t-s)^{-\frac{1}{2}}-t^{-\frac{1}{2}}= (t-s)^{-\frac{1}{2}} t^{-\frac{1}{2}}\frac{s}{(t-s)^{\frac{1}{2}}+t^{\frac{1}{2}}}\leq t^{-1} (t-s)^{-\frac{1}{2}}s.$$ A direct computation then shows that it is of order $t^{-\frac{1}{2}-\delta}.$

Hence ~\eqref{eq:secondStep} is proved.
\begin{flushright}
 $\square$
\end{flushright}


\subsection{Best Decay Estimate}\label{SEC:bestconstants}
Here we consider the decay estimate for the solution of the linear integro-differential equation ~\eqref{eq:integro}.
Recall that $Z$ is a positive constant, $W:\mathbb{R}^3\rightarrow \mathbb{R}$ is a spherically symmetric and of rapid decay.

We consider two possibilities: The solution $q$ decays faster and slower than $t^{-1},$ respectively. In the first case we have the following result.
\begin{theorem}\label{THM:integrable}
Suppose that the solution $q:\ \mathbb{R}^{+}\rightarrow \mathbb{R}$ of ~\eqref{eq:integro} satisfies the asymptotic form
\begin{equation}\label{eq:integrable}
q_{t}=C t^{-1-\delta}+O(t^{-1-\delta-\epsilon}),\ \text{as}\ t\rightarrow \infty,
\end{equation} for some $\delta\in (0,\frac{1}{2})$ and $\epsilon>0.$ Then
the constant $\delta$ must satisfy the equation
\begin{equation}\label{eq:best1}
\int_{0}^{\infty}e^{-t} t^{-\delta}\ dt=\frac{1}{\sqrt{\pi}(1+2\delta)}\int_{0}^{\infty} e^{-t} t^{-\frac{1}{2}-\delta}\ dt.
\end{equation}
By computer simulation there exists exactly one solution $\delta$, which belongs to the interval $(0.15,0.17)$.
\end{theorem}

In the second case we have the following result.
\begin{theorem}\label{THM:nonIntegrable}
There does not exists $\delta\in (0,\frac{1}{2})$ such that the solution $q:\ \mathbb{R}^{+}\rightarrow \mathbb{R}$ of Eq.~\eqref{eq:integro} has the asymptotic form
\begin{equation}\label{eq:NonIntegrable}
q_{t}=C t^{-\frac{1}{2}-\delta}+O(t^{-\frac{1}{2}-\delta-\epsilon}),
\end{equation} for some $\epsilon>0.$ Equivalently, there is no $\delta\in (0,\frac{1}{2})$ such that the following equation holds:
\begin{equation}\label{eq:best2}
\frac{1}{\delta (1-2\delta)\sqrt{\pi}} \int_{0}^{\infty} e^{-t} t^{-\delta}\ dt= \int_{0}^{\infty} e^{-t} t^{-\frac{1}{2}-\delta}\ dt.
\end{equation}
\end{theorem}

The two theorems will be proven in Sections ~\ref{sec:integrable} and ~\ref{sec:nonIntegrable}.

The general ideas of the proofs are simple. We Fourier-transform both sides of ~\eqref{eq:integro} and look for possible values of $\delta$ consistent with that equation.

To prepare for the proof we begin with deriving a convenient expression for $q_t$. By Fourier transformation we find that
\begin{itemize}
\item[(A)]
\begin{equation}
\int_{0}^{\infty} e^{ikt} \dot{q}_{t}\ dt=-1-ik\int_{0}^{\infty}e^{ikt} q_t\ dt.
\end{equation}
\item[(B)]
\begin{equation}
-\int_{0}^{\infty} e^{ikt} Re \langle W,\ \int_{0}^{t}  e^{i\frac{\Delta}{2}(t-s)}  q_{s}\ ds\ W\rangle dt= G(k+i0)\int_{0}^{\infty} e^{ikt} q_t\ dt
\end{equation} where $G(\cdot+i0):\mathbb{R}\rightarrow \mathbb{C}$ is the function defined in ~\eqref{eq:difGk}.
\end{itemize}

Define a function $\Psi:\ \mathbb{R}\rightarrow \mathbb{C}$ by
\begin{equation}\label{eq:difPsi}
\Psi(k):=\int_{0}^{\infty} e^{ikt} \ Re\langle W, \  e^{i\frac{\Delta}{2}t} W\rangle \int_{0}^{t} q_{s}\ dsdt.
\end{equation}
Then the computations above and ~\eqref{eq:integro} imply that
\begin{equation*}
\int_{0}^{\infty} e^{ikt} q_{t}\ dt=-\frac{1+Z \Psi(k)}{ik+Z G(k+i0)}.
\end{equation*}
To simplify this equation the following observation is useful.
\begin{lemma}
If $|q_{t}|\leq c(1+t)^{-\frac{1}{2}-\epsilon}$ for some positive constants $c$ and $\epsilon$, then we have that
\begin{equation}
1+Z \Psi(0)=0.
\end{equation}
\end{lemma}
\begin{proof}
By the definition of $\Psi$ it is enough to prove that $$\int_{0}^{T}Re \langle W,\ \int_{0}^{t}  e^{i\frac{\Delta}{2}(t-s)}  q_{s}\ ds\ W\rangle dt\rightarrow 0,$$ as $T\rightarrow \infty.$ To see this we use simple integration by parts in the variable $t$ to find $$\int_{0}^{T}Re \langle W,\ \int_{0}^{t}  e^{i\frac{\Delta}{2}(t-s)}  q_{s}\ ds\ W\rangle dt=Re\langle W, \ e^{i\frac{\Delta}{2}  T}(i\Delta)^{-1} W\rangle \int_{0}^{T} q_s\ ds.$$ The lemma follows from the observation that $\langle W, \ e^{i\frac{\Delta}{2}  T}(i\Delta)^{-1} W\rangle=O(T^{-\frac{1}{2}})$.
\end{proof}
This result implies that
\begin{equation}\label{eq:consistency}
\int_{0}^{\infty} e^{ikt} q_{t}\ dt=-\frac{Z [\Psi(k)-\Psi(0)]}{ik+Z G(k+i0)}.
\end{equation}

In the following we prove the two theorems, using the equation ~\eqref{eq:consistency}.
\subsubsection{\large{Proof of Theorem ~\ref{THM:integrable}}}\label{sec:integrable}
We assume that $k>0.$ This suffices thanks to that the positive $k$ and the negative $k$ differ just a complex conjugation of ~\eqref{eq:consistency}.

We first analyze the term $\Psi$ in ~\eqref{eq:consistency}.
By direct computation
\begin{equation}\label{eq:PsiInte}
\begin{array}{lll}
\Psi(k)&=&\int_{0}^{\infty} e^{ikt} Re\langle W,  e^{i\frac{\Delta}{2}t} W\rangle \int_{0}^{t} q_s\ ds dt\\
& &\\
&=& \int_{0}^{\infty} e^{ikt} Re\langle W,  e^{i\frac{\Delta}{2}t} W\rangle  dt\ \int_{0}^{\infty} q_s\ ds-\int_{0}^{\infty} e^{ikt} Re\langle W,  e^{i\frac{\Delta}{2}t} W\rangle \int_{t}^{\infty} q_s\ ds dt
\end{array}
\end{equation} Notice $$\int_{0}^{\infty} e^{ikt} Re\langle W,  e^{i\frac{\Delta}{2}t} W\rangle  dt=-G(k+i0)$$ where the function $G$ is defined in ~\eqref{eq:difGk}.

Next, observe that part of the integrand, $Re\langle W,  e^{i\frac{\Delta}{2}t} W\rangle \int_{t}^{\infty} q_s\ ds,$ in the second term is integrable on $[0,\infty).$
\begin{equation}
\begin{array}{lll}
& &-\int_{0}^{\infty} e^{ikt} Re\langle W,  e^{i\frac{\Delta}{2}t} W\rangle \int_{t}^{\infty} q_s\ ds dt\\
& &\\
&=&\int_{0}^{\infty}e^{ikt} \partial_{t} [\int_{t}^{\infty} Re\langle W, e^{i\frac{\Delta}{2}  s}W\rangle \int_{s}^{\infty} q_{s_{1}}\ ds_{1} ds]\ dt\\
& &\\
&=&-  \int_{0}^{\infty} Re\langle W, e^{i\frac{\Delta}{2}  s}W\rangle \int_{s}^{\infty} q_{s_{1}}\ ds_{1} ds-ik\int_{0}^{\infty}e^{ikt} [\int_{t}^{\infty} Re\langle W, e^{i\frac{\Delta}{2}  s}W\rangle \int_{s}^{\infty} q_{s_{1}}\ ds_{1} ds]\ dt
\end{array}
\end{equation}
The second term on the right hand side vanishes when $k=0.$  This, together with the fact $G(0+i0)=0$ and ~\eqref{eq:PsiInte}, implies that $$-  \int_{0}^{\infty} Re\langle W, e^{i\frac{\Delta}{2}  s}W\rangle \int_{s}^{\infty} q_{s_{1}}\ ds_{1} ds=\Psi(0).$$

Collecting these estimates we find that
\begin{equation}\label{eq:asyPsi}
\Psi(k)-\Psi(0)=-G(k+i0)\int_{0}^{\infty} q_s\ ds+Y(k),
\end{equation} where the function $Y$ is defined by
$$Y(k):=-ik\int_{0}^{\infty}e^{ikt} [\int_{t}^{\infty} Re\langle W, e^{i\frac{\Delta}{2}  s}W\rangle \int_{s}^{\infty} q_{s_{1}}\ ds_{1} ds]\ dt.$$

Expressions ~\eqref{eq:asymp} show that
$$-\int_{t}^{\infty} Re\langle W, e^{i\frac{\Delta}{2}  s}W\rangle \int_{s}^{\infty} q_{s_{1}}\ ds_{1} ds
=C\frac{2}{\delta}\pi^{\frac{3}{2}}\frac{1}{\frac{1}{2}+\delta}t^{-\frac{1}{2}-\delta}+t^{-\frac{1}{2}-\delta-\epsilon},\ \text{as}\ t\rightarrow \infty$$ where the constants $C$ and $\epsilon$ are the same as those in Theorem ~\ref{THM:integrable}.

For positive, and small $k$, the function $Y$ in ~\eqref{eq:asyPsi} has the form
\begin{equation*}
Y(k)=C\frac{i }{\frac{1}{2}+\delta}\frac{2}{\delta}\pi^{\frac{3}{2}} k^{\frac{1}{2}+\delta} \int_{0}^{\infty} e^{it} t^{-\frac{1}{2}-\delta}\ dt+correction.
\end{equation*}
By deformation of the integration contour we find that
\begin{equation}\label{eq:contour}
\int_{0}^{\infty} e^{it} t^{-\frac{1}{2}-\delta}\ dt=\int_{0}^{i\infty} e^{it} t^{-\frac{1}{2}-\delta}\ dt=i^{\frac{1}{2}-\delta}\int_{0}^{\infty} e^{-t} t^{-\frac{1}{2}-\delta}\ dt.
\end{equation} Hence
\begin{equation}\label{eq:asymY}
Y(k)=-\frac{C}{\frac{1}{2}+\delta}\frac{2}{\delta} \pi^{\frac{3}{2}}  i^{-\frac{1}{2}-\delta} k^{\frac{1}{2}+\delta} \int_{0}^{\infty} e^{-t} t^{-\frac{1}{2}-\delta}\ dt+correction.
\end{equation}

For the term on the left hand side of ~\eqref{eq:consistency} we use the assumption in Theorem ~\ref{THM:integrable} to arrive at
\begin{equation*}
\begin{array}{lll}
 \int_{0}^{\infty} e^{ikt}q_t\ dt
&=& -\int_{0}^{\infty} e^{ikt} \partial_{t}\int_{t}^{\infty}q_s \ dsdt\\
& &\\
&=&\int_{0}^{\infty} q_s \ ds+ ik \int_{0}^{\infty} e^{ikt} \int_{t}^{\infty}q_s\ dsdt
\end{array}
\end{equation*}
By the assumption on $q$ in Theorem ~\ref{THM:integrable}, we find
\begin{equation}\label{eq:rightHand}
\begin{array}{lll}
\int_{0}^{\infty} e^{ikt}q_t\ dt
&=& \int_{0}^{\infty}q_s \ ds+C\frac{ik}{\delta} \int_{0}^{\infty}e^{ikt} t^{-\delta} \ dt+correction\\
& &\\
&=&\int_{0}^{\infty}q_s \ ds- C\frac{i^{-\delta}k^{\delta}}{\delta}\int_{0}^{\infty} e^{-s} s^{-\delta}\ ds+correction,
\end{array}
\end{equation} where in the last step we deform the contour of integration as in ~\eqref{eq:contour}.

We now return to ~\eqref{eq:consistency}. Observe that the term $\int_{0}^{\infty}q_s\ ds$ appears on both sides. Hence they cancel each other. Next, we compare the terms of order $k^{\delta}$. Recall that $G(k+i0)=(i-1)\pi^2 k^{\frac{1}{2}}+O(k)$ in ~\eqref{eq:Gki0}. This, together with ~\eqref{eq:asyPsi}, ~\eqref{eq:asymY} and ~\eqref{eq:rightHand}, implies that if ~\eqref{eq:consistency} holds then we must have that
$$(i-1) \int_{0}^{\infty}e^{-t} t^{-\delta}\ dt+i^{-\frac{1}{2}}\frac{1}{\sqrt{2\pi}(\frac{1}{2}+\delta)}\int_{0}^{\infty} e^{-t} t^{-\frac{1}{2}-\delta}\ dt
=0.$$

Using that $i^{\frac{1}{2}}=\frac{1}{\sqrt{2}}(1+i)$, we conclude that Theorem ~\ref{THM:integrable} holds.

\subsubsection{\large{Proof of Theorem ~\ref{THM:nonIntegrable}}}\label{sec:nonIntegrable}
We start with Equation ~\eqref{eq:consistency}. It is enough to consider the case $k>0.$

Recall the definition of $\Psi$ in ~\eqref{eq:difPsi}. By direct computation
\begin{equation}\label{eq:PsiNon}
\begin{array}{lll}
\Psi(k)&=&-\int_{0}^{\infty} e^{ikt} \partial_{t} \int_{t}^{\infty} Re\langle W, e^{i\frac{\Delta}{2}  s}W\rangle \int_{0}^{s} q_{s_1}\ ds_1\ dsdt\\
& &\\
&=&\int_{0}^{\infty} Re\langle W, e^{i\frac{\Delta}{2}  s}W\rangle \int_{0}^{s} q_{s_1}\ ds_1\ ds+\int_{0}^{\infty} ik e^{ikt} \int_{t}^{\infty}  Re\langle W, e^{i\frac{\Delta}{2}  s}W\rangle \int_{0}^{s} q_{s_1}\ ds_1\ dsdt
\end{array}
\end{equation}
We observe that the second term on the right hand side vanishes at $k=0$. Hence the first term, which is a constant, must be $\Psi(0).$

We now evaluate the second term. The assumption on $q$ in Theorem ~\ref{THM:nonIntegrable}
and the asymptotic form for $Re\langle W,  e^{i\frac{\Delta}{2}t} W\rangle$ in ~\eqref{eq:asymp} imply that
\begin{equation*}
Re\langle W,  e^{i\frac{\Delta}{2}t} W\rangle \int_{0}^{t}q_s\ ds=-\frac{2C}{\frac{1}{2}-\delta}  \pi^{\frac{3}{2}}  t^{-1-\delta}+O(t^{-1-\delta-\epsilon}),\ \text{as}\ t\rightarrow \infty;
\end{equation*} hence
\begin{equation}
\int_{t}^{\infty}Re\langle W, e^{i\frac{\Delta}{2}  s}W\rangle \int_{0}^{s}q_{s_1}\ ds_{1}ds=-\frac{2C}{\delta (\frac{1}{2}-\delta)}
\pi^{\frac{3}{2}} t^{-\delta}+O(t^{-\delta-\epsilon}).
\end{equation}
Plugging this into ~\eqref{eq:PsiNon} we find that, for $k>0$ small,
\begin{equation}\label{eq:Psik0}
\begin{array}{lll}
\Psi(k)-\Psi(0)
&=&-\frac{2iC}{\delta (\frac{1}{2}-\delta)} \pi^{\frac{3}{2}} k^{\delta} \int_{0}^{\infty} e^{it} t^{-\delta} \ dt+correction\\
& &\\
&=&\frac{2C}{\delta (\frac{1}{2}-\delta)} i^{-\delta} \pi^{\frac{3}{2}} k^{\delta} \int_{0}^{\infty} e^{-t} t^{-\delta} \ dt+correction.
\end{array}
\end{equation}

For the term on the left hand side of ~\eqref{eq:consistency}
\begin{equation}\label{eq:FouSec}
\begin{array}{lll}
\int_{0}^{\infty} e^{ikt} q_t \ dt
&=&  Ck^{-\frac{1}{2}+\delta} \int_{0}^{\infty} e^{it} t^{-\frac{1}{2}-\delta}\ dt+correction\\
& &\\
&=&C i^{\frac{1}{2}-\delta}  k^{-\frac{1}{2}+\delta}\int_{0}^{\infty} e^{-t} t^{-\frac{1}{2}-\delta}\ dt+correction
\end{array}
\end{equation}

Identity ~\eqref{eq:consistency} and ~\eqref{eq:Psik0}, ~\eqref{eq:FouSec}, ~\eqref{eq:Gki0} imply that
\begin{equation*}
\frac{1}{\delta (\frac{1}{2}-\delta)} \int_{0}^{\infty} e^{-t} t^{-\delta}\ dt-2\pi^{\frac{1}{2}} \int_{0}^{\infty} e^{-t} t^{-\frac{1}{2}-\delta}\ dt=0,
\end{equation*} which is Theorem ~\ref{THM:nonIntegrable}.



\end{document}